\documentclass[draftclsnofoot, onecolumn, comsoc, 12pt]{IEEEtran} 

\def\argmax{\mathop{\rm arg\,max}}

\usepackage{graphicx,epsfig}
\usepackage[noadjust]{cite}
\usepackage{mcite}
\usepackage{amsfonts,helvet}
\usepackage{fancyhdr}
\usepackage{threeparttable}
\usepackage{epsf}
\usepackage{amsthm}
\usepackage{amsmath}
\usepackage{siunitx}
\usepackage{amssymb}
\usepackage{booktabs}

\usepackage{dsfont}
\usepackage{subfigure}
\usepackage{color}
\usepackage[linesnumbered,ruled, noend]{algorithm2e}
\usepackage{algpseudocode}
\usepackage{algcompatible}
\usepackage{enumerate}
\usepackage{cancel}

\newtheorem{lemma}{Lemma}

\newtheorem{proposition}{Proposition}
\newtheorem{remark}{Remark}

\usepackage{eucal}

\def\bff{{\bf f}}
\def\bg{{\bf g}}
\def\bh{{\bf h}}

\def\bs{{\bf s}}

\def\bv{{\bf v}}
\def\bw{{\bf w}}
\def\bx{{\bf x}}

\def\bz{{\bf z}}

\def\bA{{\bf A}}
\def\bB{{\bf B}}
\def\bC{{\bf C}}
\def\bD{{\bf D}}

\def\bF{{\bf F}}

\def\bH{{\bf H}}
\def\bI{{\bf I}}

\def\bR{{\bf R}}

\def\bW{{\bf W}}

\def\bbE{\mbox{$\mathbb{E}$}}

\begin{document}


\title{ Joint Precoding and Artificial Noise Design \\ for MU-MIMO Wiretap Channels}

\author{
Eunsung~Choi, Mintaek Oh, Jinseok Choi,\\ Jeonghun~Park, Namyoon Lee, and Naofal Al-Dhahir
\thanks{E. Choi, M. Oh, and J. Choi are with the Department of Electronical Engineering, Ulsan National Institute of Science and Technology (UNIST), Ulsan, South Korea (e-mail: {\texttt{\{eunsungchoi,ohmin,jinseokchoi\}@unist.ac.kr}}). 



N. Lee is with the School of Electrical Engineering, Korea University, Seoul, South Korea (e-mail: {\texttt{namyoon@korea.ac.kr}}). 

J. Park is with the School of Electronics Engineering, College of IT Engineering, Kyungpook National University, Daegu, South Korea (e-mail: {\texttt{jeonghun.park@knu.ac.kr}}). 

N. Al-Dhahir is with the Electrical and Computer Engineering Department, The University of Texas at Dallas, TX, USA. 
(e-mail: {\texttt{aldhahir@utdallas.edu}}).
}
}

\maketitle
\setcounter{page}{1} 

\begin{abstract} 

Secure precoding superimposed  with artificial noise (AN) is a promising transmission technique to improve security by harnessing the superposition nature of the wireless medium. However, finding a jointly optimal precoding and AN structure is very challenging in downlink multi-user multiple-input  multiple-output (MU-MIMO) wiretap channels with multiple eavesdroppers. The major challenge in maximizing the secrecy rate arises from the non-convexity and non-smoothness of the rate function. Traditionally, an alternating optimization framework that identifies beamforming vectors and AN covariance matrix has been adopted; yet this alternating approach has  limitations in maximizing the secrecy rate. In this paper, we put forth a novel secure precoding algorithm that jointly and simultaneously optimizes the beams and AN covariance matrix for maximizing the secrecy rate when a transmitter has either perfect or partial channel knowledge of eavesdroppers. To this end, we first establish an approximate secrecy rate in a smooth function. Then, we derive the first-order optimality condition in the form of the nonlinear eigenvalue problem (NEP). We present a computationally efficient algorithm to identify the principal eigenvector of the NEP as a suboptimal solution for secure precoding. Simulations demonstrate that the proposed methods improve secrecy rate significantly compared to the existing secure precoding methods.

\end{abstract}

\section{Introduction}


As the amount of information delivered via wireless media grows rapidly, security in wireless communications has become an increasingly critical issue to protect confidential information from eavesdroppers.
Cryptography has been used in traditional approaches to protect confidential information from eavesdropping  \cite{massey:proc:88}.
Unfortunately, the key distribution and complicated encryption algorithms necessitate high implementation costs.
Accordingly, systems with low-computing capability, for instance, internet-of-things (IoT) communication systems and sensor networks, are not capable of using the cryptography approach to prevent eavesdropping the signals transmitted by an access point (AP).

Compared to the cryptography approach, physical layer security  has been flourishing due to its low computational complexity \cite{wyner:bell:75}.
It has been demonstrated in \cite{gopala:tit:08, liang:tit:08, khisti:tit:08, tekin:tit:08} that a transmitter can reliably deliver confidential information to users with a positive rate while guaranteeing that an eavesdropper cannot decode the information if the wiretap channel quality is worse than the legitimate channel.
Then, {\emph{the secrecy rate}} is defined as this non-zero transmission rate.
The importance of security has become more  significant as the amount of transmitted information is rapidly increasing.
Consequently, there are many works that apply physical layer security for the secrecy rate maximization in 6G applications such as IoT systems  \cite{jaiswal2020secrecy, xu2016security, haider2020optimization}.
In physical layer security, there are two representative approaches which solve the secrecy rate maximization problem in multiple-input multiple-output (MIMO) systems: $(i)$ secure precoding and $(ii)$ artificial noise (AN) injection.
In this paper, we investigate a secure precoding with AN design for physical layer security of downlink IoT networks and propose a joint optimization framework.

\subsection{Prior Work}




Starting with Wyner's pioneering work \cite{wyner:bell:75}, physical layer security has been studied in many scenarios including multiple access channels \cite{liang:tit:08} and broadcast channels \cite{khisti:tit:08}. 
In particular, there have been extensive efforts to evaluate the secrecy rate of multiple antenna systems employing an information-theoretic approach.
The secrecy rate was derived in \cite{parada:isit:05} for single-input multiple-output (SIMO) channels in a slow fading environment.
A single-antenna eavesdropper \cite{shafiee:isit:07} and a multiple antennas eavesdropper \cite{khisti:tit:10_part1} were also investigated to analyze the achievable secrecy rate in multiple-input single-output (MISO) channels.
Furthermore, the secrecy rate was derived under a multiple-input multiple-output (MIMO) assumption in \cite{khisti:iti:10, oggier:tit:11}.
Specifically, the generalized singular-value decomposition (GSVD) was proposed in \cite{khisti:iti:10} to achieve the MIMO channels' secrecy rate in the presence of a multi-antenna eavesdropper.
Secure two-user MIMO broadcast channels were considered in \cite{liu:tit:09, liu:tit:10}, and secret dirty-paper coding (S-DPC) was proposed to achieve the secrecy rate region.
In \cite{fakoorian:jsac:13}, it was also revealed that linear precoding can attain the same secrecy rate region as S-DPC. 
In \cite{fakoorian:tifs:11}, the achievable secrecy rates and linear precoding design in MIMO interference channels were investigated.
In \cite{kampeas:isit:16, wang:tifs:17}, the secrecy rate was investigated when multiple users and multiple eavesdroppers coexist.

There have been several prior works that proposed secure precoding method to maximize the secrecy rate.
In \cite{li:tvt:18}, a secure beamforming approach was proposed to minimize the signal-to-interference plus noise (SINR) of eavesdroppers in a multiuser MIMO (MU-MIMO) system with a single eavesdropper overhearing information of a particular user. 
Assuming MIMO multi-eavesdropper (MIMOME) channels, robust beamforming methods were proposed in \cite{mukherjee:tsp:11, 9312431, choi:wcl:20} where the channel state information (CSI) is imperfect.
In \cite{zhao:tifs:15}, a successive convex approximation technique was introduced to relax the sum secrecy rate maximization problem.
For wireless systems with multiple transmitter-receiver pairs and one eavesdropper, algorithms were proposed in \cite{sheng:tsp:18, sheng:spl:18} to maximize the secrecy rate and the secrecy energy efficiency.
In \cite{yang:tcom:13, yang:tcom:16}, a secure antenna selection method was investigated to improve the secrecy performance without increasing computational complexity.
Friendly jamming methods were studied in \cite{yang:twc:13, ma:tcom:18, gao:tcom:19} where the transmitter cannot estimate the CSI of the eavesdroppers.
To consider the secrecy rate and the secrecy energy efficiency together, massive MIMO systems \cite{nguyen:jsac:18, yang:tifs:19}, low-resolution digital-to-analog converter systems \cite{xu:tcom:19, xu:access:19}, and simultaneous wireless information and power transfer  systems \cite{nasir:twc:17, alageli:jsac:19, sun:wcl:19, tang:access:18} have been investigated.
A comprehensive examination of physical layer security in multiuser systems was performed in \cite{mukherjee:surv:14}.
In \cite{zheng:twc:17}, the hybrid full- and half-duplex receiver deployment strategy was studied in a wireless ad hoc network with numerous legitimate transmitter-receiver pairs and eavesdroppers.

The AN injection scheme is another representative approach for achieving secure communications in MIMO systems.
In particular, AN design based on the null-space of the legitimate channel matrix is a well-known methodology for maximizing the sum secrecy rate \cite{1558439,7328729}.
In \cite{yu:access:18}, assuming that there is a single user and a single eavesdropper, the  design of secure precoding and AN covariance  was presented to maximize the secrecy rate.
In \cite{9201173}, an AN-aided secrecy rate maximization method was presented for the system with an intelligent reflecting surface. 
In \cite{6482662}, the secure precoder for the MISO scenario was found by solving the AN-aided secrecy rate maximization problem using the convex optimization toolbox, CVX. 
In \cite{mei2017artificial}, AN-aided transmit design for multi-user MISO systems was proposed to solve the secrecy rate region maximization (SRRM) problem.
For the MIMO scenario, optimal transmit power allocation for precoder and AN was considered in \cite{6572838}.
The AN-aided scheme was also proposed in \cite{nguyen:jsac:18} where AN is injected into the downlink training signals to block the eavesdropper from obtaining CSIT of the correct wiretap link in the presence of a single-user and multiple eavesdroppers.
The AN-aided secure beamforming approach for data transmission was investigated also for the single-user and multi-eavesdropper model in  \cite{5643181} and for the multi-group and multi-cast MU-MIMO model in \cite{8319269}. 
Furthermore, the power ratio between the secure precoder and AN covariance matrix was optimized by using the alternating optimization approach.
In \cite{nguyen:jsac:18,7605477}, the secrecy rate maximization method was investigated with the secure precoder and fixed null-space AN by alternation-based power optimization.
The power optimization between secure precoder and fixed null-space AN was performed by a one-dimensional line search for maximizing the secrecy rate in \cite{6482662}.
In \cite{5306434}, the proposed non-adaptive power allocation method was performed between a fixed zero-forcing (ZF) precoder and a fixed null-space AN. 
In \cite{qin2011optimal}, AN-aided secure communication design with the fixed ZF precoding method was researched by the proposed two-level optimization method.

Although physical layer security has been widely investigated, the existing secure precoding and AN injection approaches have the limited applicability for general downlink broadcasting wiretap channels as follows: $(i)$ only a single legitimate user is assumed in the wiretap channel \cite{yu:access:18,9201173,6572838,6482662,5643181},
$(ii)$ only a single eavesdropper presents to overhear confidential messages \cite{mukherjee:tsp:11, li:tvt:18, xu:access:19, yang:tifs:19, xu:tcom:19, yu:access:18,9201173,nguyen:jsac:18,6572838},
$(iii)$ there are fixed pairs of an eavesdropper and legitimate user, i.e., each user has one designated eavesdropper \cite{zhao:tifs:15}, which is not a general wiretap channel where each eavesdropper is considered to overhear any legitimate user's message, 
$(iv)$ existing algorithms mostly require high computational complexity as in \cite{nguyen:jsac:18,alageli:jsac:19,6482662,8319269,mei2017artificial},
or $(v)$ optimization of precoding and AN  covariance is performed in an alternating manner \cite{nguyen:jsac:18,6482662,7605477,5306434,qin2011optimal}, which is highly sub-optimal.
Such limitations for physical layer security comes from the complicated nature of the optimization problem for AN-aided secure precoding; the sum rate maximization problem is a famous NP-hard problem even if there are only users in the network \cite{choi:arxiv:19} and thus, solving a secure precoding problem for wiretap channels is more challenging.
Furthermore, for the multiple eavesdropper case, the amount of information leakage is determined by the maximum rate of each wiretap channel \cite{kampeas:isit:16} which causes the non-smoothness of the optimization problem.
Finally, the design philosophies of the precoder and AN covariance are disparate as the AN does not convey any information to legitimate users, which leads to alternation-based optimization for most AN-aided secure precoding methods.
In this regard, it is necessary to develop a joint and simultaneous optimization framework of secure precoding and AN covariance design for general downlink broadcasting wiretap channels where multiple users and eavesdroppers coexist.
Overcoming the forementioned challenges, we propose a joint optimization framework without alternation for the system, i.e., joint and simultaneous optimization.

\subsection{Contributions}

In this paper, we propose a joint optimization framework for optimizing the secure precoder and AN covariance simultaneously in downlink systems where multiple legitimate users and multiple eavesdroppers coexist. 
In the considered system, a multi-antenna AP transmits confidential information symbols via linear precoding to  $K$ single-antenna users.
In addition, coexisting $M$ single-antenna eavesdroppers attempt to overhear the confidential information signals sent from the AP.
In this setup, the AP employs both precoding and AN  to maximize the sum secrecy rate.
Here, we summarize our contributions.

\begin{itemize}
\item We adopt a secrecy rate as our key performance measure.
Using the secrecy rate, we formulate a sum secrecy rate maximization problem  to jointly optimize $(i)$  precoding matrix and $(ii)$ an AN covariance matrix.
There are critical difficulties in solving the formulated problem.
First, the problem is non-convex and hence, finding the globally optimal solution is infeasible.
Second, the sum secrecy rate is non-smooth since the secrecy rate in the presence of multiple eavesdroppers is determined by the maximum wiretap channel rate of the eavesdroppers.
Finally, the precoder and AN are coupled in the sum secrecy rate term with different design directions, which hinders the joint optimization, leading to an alternating optimization approach eventually. 

\item To address these challenges, we first approximate the non-smooth objective function by using a smooth maximum approximation.
Then, concatenating the precoding vectors and AN covariance vectors, the original problem is reformulated into a tractable non-convex form which is a product of Rayleigh quotients.
A significant advantage of this reformulation is that the first-order Karush-Kuhn-Tucker (KKT) condition of the reformulated problem can be cast as a generalized eigenvalue problem.
In particular, the corresponding eigenvalue is equivalent to the objective function of the problem and the eigenvector is the normalized concatenated vector.
Consequently, finding the leading eigenvector of the generalized eigenvalue problem is equivalent to finding a stationary point that maximizes the objective function, i.e., the best local optimum.
\item We propose an algorithm to efficiently find the leading eigenvector based on generalized power iteration (GPI) \cite{choi:arxiv:19}.  
Accordingly, the proposed GPI-based optimization method jointly optimizes secure precoding and AN covariance without alternation. 
In other words, we obtain  secure precoding vectors, AN covariance matrix, and the power ratio between the precoder and AN simultaneously from the GPI-based method.
In addition, thanks to the block-diagonal structure of the matrices of the generalized eigenvalue problem, the algorithm complexity can achieve low complexity.
\item Considering the imperfect CSIT for wiretap channels, we further develop a channel covariance-based optimization method by exploiting a similar joint optimization framework to the perfect CSIT case.
Since designing the AN covariance as the null-space of user channels, we also adopt the null-space AN design approach to the proposed framework and propose a null-space AN-based secure precoding method.


\item Through simulations, we validate the secrecy performance of the proposed joint and simultaneous optimization method.
We demonstrate that the proposed method improves secrecy performance significantly in a various system setups compared to the other baseline methods thanks to the joint optimization without alternation.
\end{itemize}

\begin{table}[!t]\caption{A List of Key Notations }
\begin{center} 
    \begin{tabular}{ c | p{6 cm} || c|  p{6 cm} } 
    \toprule
    $N$ & number of AP antennas &  $(\cdot)^{\sf T}$ & matrix transpose\\
    $K$ & number of users & $(\cdot)^{\sf H}$ & matrix Hermitian \\
    $M$ &   number of eavesdroppers  &  $(\cdot)^{*}$ & complex conjugate \\
    $P$ &  transmit power constraint &  $(\cdot)^{-1}$ & matrix inversion \\
    $\bF$ & precoding matrix &${\sf{tr}}(\cdot)$ & trace operator \\
    $P\bf \Phi\bf \Phi^{\sf H}$ &  AN covariance matrix & 
    ${\rm vec}(\cdot)$ & vectorization operator\\
    $\bh_k$ &  legitimate channel vector from $k$-th user & $\| \cdot \|$ &  $\ell_2$ norm \\
    $\bg_m$ &  wiretap channel vector of $m$-th eavesdropper &  $\left[\cdot\right]^+$ &  $\max\left[\cdot, 0\right]$ \\
    $R_k$ & rate of user $k$  &  ${\bf{I}}_N$ & $N \times N$ identity matrix  \\ 
    $R_{m,k}^{\sf e}$ & rate of eavesdropper $m$ for $k$-th symbol & $\bf 0$ & zero vector with proper dimension\\ 
   \bottomrule
    \end{tabular}
    \vspace{-2em}
\end{center}
\end{table}

\section{System Model} \label{sec:sys_model}

In this section, we describe the considered  downlink  communications system.
We focus on the scenario of multiple users in which an AP with $N$ antennas supports $K$ single-antenna users.
In addition, $M$ single-antenna eavesdroppers coexist, attempting to overhear the private messages of the users. 
We let $\CMcal{K}$ and $\CMcal{M}$ denote the set of users and eavesdroppers, respectively.
The AP precodes user symbols $s_k$, $k=1,...,K$, with a linear precoder $\bff_k$ and transmits the precoded signals along with an AN to enhance communication security as
\begin{align} \label{eq:x_org}
    \bx = \bF\bs+{\bf \Phi} \bz,
\end{align}
where ${\bF}=[\bff_1,...,\bff_K]\in\mathbb{C}^{N \times K}$ is the precoding matrix, $\bs=[s_1,...,s_K]^{\sf T}$ is the vector of user symbols with $\bs\sim\mathcal{CN}(\bf {0}_{\textit{N} \times \textit{1}}, \textit{P} \bI_{\textit {N}})$, and ${\bf \Phi}\bz$ is the AN vector which follows ${\bf \Phi}\bz\sim\mathcal{CN}(\bf{0}, \textit{P} {\bf \Phi} {\bf \Phi}^{\sf H})$.
Specifically, we utilize the $J$ number of columns to design AN covariance matrix, i.e., $\bf \Phi$ is a matrix of size $N\times J$.
Accordingly, the total transmit power constraint is given by 
\begin{align}
    P \sum_{i=1}^{K} \| {\bf{f}}_i \|^2 + P \sum_{j=1}^{J} \| {\pmb {\phi}}_j \|^2    \leq P
\end{align} 
where ${\pmb {\phi}}_j$ is the $j$-th column of $\bf \Phi$ and $P$ is maximum transmit power constraint.

The legitimate channel vector from $k$-th user to the the AP is denoted by ${\bf{h}}_{k} \in \mathbb{C}^{N}$ for $k\in \CMcal{K}$.
The spatial covariance matrix ${\bf{R}}_k \in \mathbb{C}^{N \times N}$ of the channel vector ${\bf{h}}_k$ is defined as ${\bf{R}}_{k} = \mathbb{E}\left[{\bf{h}}_{k}  {\bf{h}}_{k}^{\sf H} \right] $.
Similarly, the wiretap channel vector from the $m$-th eavesdropper to the AP is denoted by ${\bf{g}}_{m} \in \mathbb{C}^{N}$
with the covariance of ${\bf{R}}_{m}^{\sf e} = \mathbb{E}\left[ {\bf{g}}_{m} {\bf{g}}_{m}^{\sf H} \right]$.
We first assume that perfect CSIT of the users and eavesdroppers are available at the AP.
Then, we further consider the case of imperfect CSIT of eavesdroppers for a  more practical scenario where only the channel covariance of eavesdroppers is available at the AP.
The received signal at user $k$ is
\begin{align} \label{eq:rx_signal_user}
    y_{k} = \bh_k^{\sf H}\bff_k s_k + \sum_{i = 1,i\neq k}^{K} {\bf{h}}_{k}^{\sf H} {\bf{f}}_{i} s_i + \sum_{j = 1}^{J} {\bf{h}}_{k}^{\sf H} {{\pmb {\phi}}_{j} z_j} + {n_k},
\end{align}
where $n_k \sim \mathcal{CN}(0,\sigma^2)$ is additive white Gaussian noise (AWGN).
Similarly, the received signal at eavesdropper $m$ overhearing user $k$ is expressed as
\begin{align}
    y_{m,k}^{\sf e}=\bg_m^{\sf H} {\bf{f}}_{k} s_k + \sum_{i = 1, i\neq k}^{K} \bg_{m}^{\sf H} {\bf{f}}_{i} s_i + \sum_{j = 1}^{J} \bg_{m}^{\sf H} {\pmb {\phi}_j z_j} + n^{\sf e}_m,
\end{align}
where $n^{\sf e}_m \sim \mathcal{CN}(0,\sigma_{\sf e}^2)$ is the AWGN at the $m$-th eavesdropper.

In the following sections, we introduce a performance metrics and formulate a sum secrecy rate maximization problem. 
Then, we propose a novel optimization framework for joint and simultaneous precoding and AN covariance design.

\begin{figure}[t] 
    \begin{center}
    \includegraphics[width=.6\linewidth]{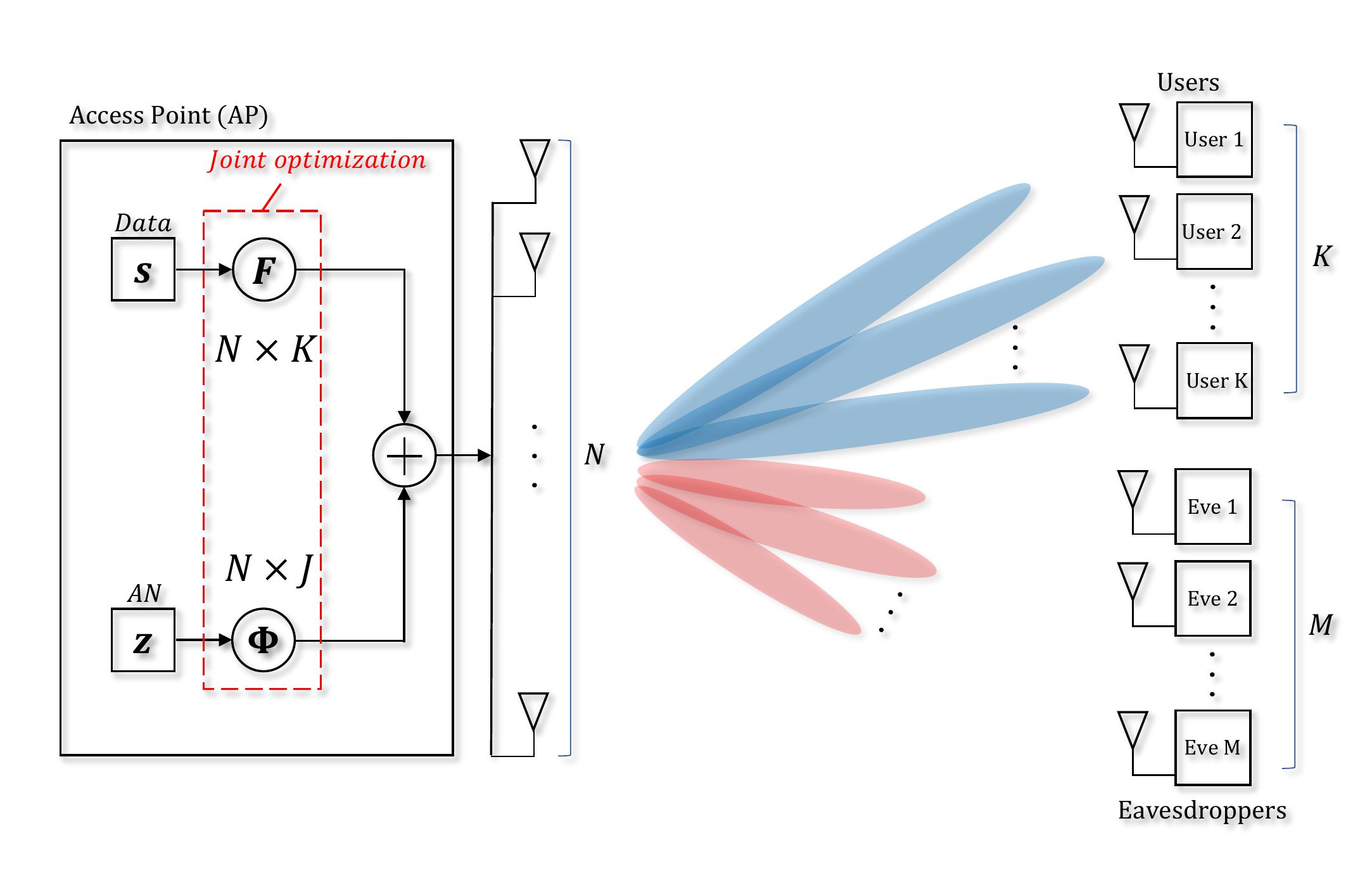}
    \end{center}
    \vspace{-2em}
    \caption{
    An AN-aided secure precoding system for downlink communications with multiple legitimate users  and multiple eavesdroppers.
    }
    \label{fig:system_model}
    
\end{figure}

\section{Problem Formulation}

We assume that each user employs single user decoding treating the interference as noise.
Consequently, the rate of user $k$ is given as
\begin{align} \label{eq:se_user}
    R_{k} & = \log_2 \left(1 + \frac{|\bh_{k}^{\sf H} \bff_k|^2}{\sum_{i = 1, i\neq k}^{K} |\bh_{k}^{\sf H} \bff_i|^2 +
    \sum_{j = 1}^{J} |{\bf{h}}_{k}^{\sf H} {\pmb {\phi}_j}|^2 + 
    \sigma^2/P}\right).
\end{align}
Similarly, the rate of the wiretap channel for the message $s_k$ achieved at eavesdropper $m$ is
\begin{align}
    \label{eq:se_eve}
    R_{m,k}^{\sf e} = \log_2 \left( 1 + \frac{|\bg_{m}^{\sf H} {\bf{f}}_k|^2}{\sum_{i = 1, i\neq k}^{K} |\bg_{m}^{\sf H} {\bf{f}}_i|^2 + \sum_{j = 1}^{J} |\bg_{m}^{\sf H} {\pmb {\phi}_j}|^2 + \sigma_{\sf e}^2/P} \right). \ \ 
\end{align}
The maximum wiretap channel rate determines the secrecy rate in the presence of multiple eavesdroppers \cite{kampeas:isit:16, yang:tcom:16}.
In this regard, we have $\max_{m \in \CMcal{M}}\left\{R_{m, k}^{\sf e}\right\}$ amount of rate loss from $R_k$ to transmit $s_k$ with guaranteed security, and the corresponding secrecy rate becomes
\begin{align}
    \left[R_k - \max_{m \in \CMcal{M}}\left\{R_{m, k}^{\sf e}\right\} \right]^{+}.
\end{align}
Then, using the secrecy rate, we formulate a sum secrecy rate maximization problem as
\begin{align}
    \label{eq:main_problem_eve}
    \mathop{{\text{maximize}}}_{\bF,\bf \Phi}& \;\; R_{\sf sum} = \sum_{k = 1}^{K} \left[R_k - \max_{m \in \CMcal{M}}\left\{R_{m, k}^{\sf e}\right\} \right]^{+}
    \\
    {\text{subject to}} & \;\; \sum_{i = 1}^{K} \left\| \bff_i \right\|^2 + \sum_{j = 1}^{J} \left\| \pmb {\phi}_j \right\|^2 \leq 1, \label{eq:main_problem_const_1}
\end{align}
where 
the constraint in \eqref{eq:main_problem_const_1} is the transmit power constraint.
We remark that to solve \eqref{eq:main_problem_eve} a joint optimization of $\bF$ and $\bf \Phi$ is necessary.

Unfortunately, the problem in \eqref{eq:main_problem_eve} is significantly challenging due to non-smoothness and non-convexity.
Even worse, the secure precoding vectors ${\bf{f}}_k$ and the AN covariance vector ${\pmb{\phi}_j}$ are coupled which makes the optimization process intractable.
In this regard, a null-space approach which fixes $\bf \Phi$ as the null-space projection matrix of the user channels was widely used \cite{1558439,7328729}.
Furthermore, alternating power allocation strategies between the secure precoder and AN covariance matrix have been actively studied \cite{7605477, 5306434}.
Unlike such well-known optimization approaches, our key contribution is the joint optimization framework proposed in the next section that determines the secure precoder, AN covariance matrix, and power ratio between precoder and AN covariance matrix without any alternation.


\section{Joint  Optimization Framework}
\label{sec:joint_optimization}

In this section, we solve the problem \eqref{eq:main_problem_eve} by resolving the key challenges through approximation and reformulation of the problem, leading to a tractable form.
Subsequently, we  propose a joint and simultaneous optimization method for the secure precoding, AN covariance matrix design, and precoder-AN power allocation for the case of perfect CSIT.
Then, we further extend the proposed framework to the case of imperfect CSIT of wiretap channels.
We also introduce the application of the proposed framework to the null-space AN covariance design approach.

\subsection{Joint and Simultaneous Optimization with Perfect CSIT} \label{sec:Opt_perfectCSIT}

To cast the secrecy rate into a tractable form, we first define a joint optimization vector by stacking all the precoding vectors $\bff_k$ and AN covariance vectors $\pmb {\phi}_j$  as
\begin{align} \label{eq:large_vector}
    \bar {\bv} = \left[\bff_1^{\sf T},\bff_2^{\sf T}, \cdots, \bff_K^{\sf T},\pmb {\phi}_1^{\sf T},\pmb {\phi}_2^{\sf T}, \cdots, \pmb {\phi}_J^{\sf T} \right]^{\sf T} \in \mathbb{C}^{N ( K+J )}.
\end{align}
We assume that the norm of \eqref{eq:large_vector} to be equal to one, i.e., $||\bar \bv||=1$, which indicates transmission at maximum transmit power since it is the optimal transmission strategy in term of transmit power and rate.
Using \eqref{eq:large_vector}, we rewrite the legitimate channel rate $R_k$ in \eqref{eq:se_user} as
\begin{align} \label{eq:Ray}
    R_k = \log_2\left( \frac{{\bf \bar v }^{\sf H} {\bf{A}}_k \bar {\bv} }{{\bf \bar v }^{\sf H} {\bf{B}}_k \bar {\bv}} \right) ,
\end{align}
where
\begin{align}
    \label{eq:A_matrix_const_user}   
    &{\bf{A}}_k = {\rm blkdiag}   \left({\bf{h}}_k {\bf{h}}_k^{\sf H}, ..., {\bf{h}}_k {\bf{h}}_k^{\sf H}  \right) + {\bf{I}}_{N   (   K+N   )} \frac{ \sigma^2}{P},
    \\
    \label{eq:B_matrix_const_user}
    &{\bf{B}}_k = {\bf{A}}_k -  {\rm blkdiag} ({\bf{0}}, \cdots, \underbrace{{\bh}_k {\bh}_k^{\sf H}}_{k{\rm th}\; {\rm block}}, \cdots, {\bf{0}} ). 
\end{align}
We note that $\bA_k$ and $\bB_k$ are block diagonal matrices with dimension of $N(K+J) \times N(K+J)$.
Similarly, we express the wiretap channel rate $R_{m,k}^{\sf e}$ in \eqref{eq:se_eve} as
\begin{align} 
    R_{m,k}^{\sf e} = \log_2\left( \frac{{\bf \bar v }^{\sf H} {\bf{C}}_{m,k} \bar {\bv} }{{\bf \bar v }^{\sf H} {\bf{D}}_{m,k} \bar {\bv}} \right),
    \label{eq:eve_Ray}
\end{align}
where
\begin{align} 
    \label{eq:C_matrix_const_eve}
    &{\bf{C}}_{m,k} = {\rm blkdiag}\left({\bg}_m {\bg}_m^{\sf H}, ..., {\bg}_m{\bg}_m^{\sf H}  \right) + {\bf{I}}_{N (K + J)} \frac{ \sigma_{\sf e}^2}{P},
    \\
    \label{eq:D_matrix_const_eve} 
    & {\bf{D}}_{m,k} = {\bf{C}}_{m,k} - {\rm blkdiag} ({\bf{0}}, \cdots, \underbrace{{\bf{g}}_m {\bf{g}}_m^{\sf H}}_{k{\rm th}\; {\rm block}}, \cdots, {\bf{0}} ). 
\end{align}
With \eqref{eq:Ray} and \eqref{eq:eve_Ray}, we reformulate the original problem \eqref{eq:main_problem_eve} to
\begin{align}
    \label{eq:reform_problem_eve}
    \mathop{{\text{maximize}}}_{\bar \bv}& \;  \sum_{k = 1}^{K}   \left\{ \log_2   \left( \frac{{\bf \bar v }^{\sf H} {\bf{A}}_k \bar {\bv} }{{\bf \bar v }^{\sf H} {\bf{B}}_k \bar {\bv}} \right)   -   \max_{m \in \CMcal{M}} \log_2   \left( \frac{\bar{\bv}^{\sf H} {\bC}_{m,k} \bar {\bv} }{\bar{\bv}^{\sf H} {\bD}_{m,k} \bar {\bv} } \right) \right\}
    \\ \nonumber
    {\text{subject to}}& \;  \| \bar {\bv} \| = 1. 
\end{align}
We remark that the reformulated problem in \eqref{eq:reform_problem_eve} is invariant up to the scaling of $\bar \bv$, i.e., scaling of  $\bar \bv$ does not change the problem in \eqref{eq:reform_problem_eve}.
As a result, we can ignore the constraint of $||\bar \bv||=1$ which has no effect on \eqref{eq:reform_problem_eve}.

Next, we utilize the following LogSumExp function to approximate the non-smooth maximum function  \cite{shen2010dual}:
\begin{align}
    \label{eq:max_logsum}
    \max_{i=1,\cdots,N}\{x_i\} \approx
    {\alpha}\log\left(\sum_{i=1}^N \exp\left(\frac{x_i}{\alpha}\right) \right), \; \; \alpha > 0,
\end{align}
where the approximation becomes tight as $\alpha \rightarrow 0$. 
Leveraging \eqref{eq:max_logsum}, we obtain the following approximation for the maximum wiretap channel rate:
\begin{align}\label{eq:max_rate}
    \max_{m \in \CMcal{M}}\left\{R_{m, k}^{\sf e}\right\} \approx \alpha\log\left(\sum_{m=1}^M \exp\left(\frac{R_{m, k}^{\sf e}}{\alpha}\right) \right).
\end{align}
Applying \eqref{eq:max_rate} to \eqref{eq:reform_problem_eve}, the sum secrecy rate $R_{\sf sum}$ is approximated as
\begin{align}
    \label{eq:max_logsum_sse}
    & R_{\sf sum} \approx \sum_{k = 1}^{K}   \left\{ \log_2 \left( \frac{{\bf \bar v }^{\sf H} {\bf{A}}_k \bar {\bv} }{{\bf \bar v }^{\sf H} {\bf{B}}_k \bar {\bv}} \right)   -   \alpha  \log\left(\sum_{m=1}^M   \exp\left(   \frac{1}{\alpha}   \log_2 \left( \frac{\bar{\bv}^{\sf H} {\bC}_{m,k} \bar {\bv} }{\bar{\bv}^{\sf H} {\bD}_{m,k} \bar {\bv} } \right) \right) \right) \right\}.
\end{align}
Finally, simplifying \eqref{eq:max_logsum_sse} further, the problem in \eqref{eq:reform_problem_eve} is reformulated as
\begin{align} \label{eq:reform_problem_apx}
    \mathop{{\text{maximize}}}_{\bar \bv}&   \sum_{k = 1}^{K}   \left\{ \log_2   \left(  \frac{{\bf \bar v }^{\sf H} {\bf{A}}_k \bar {\bv} }{{\bf \bar v }^{\sf H} {\bf{B}}_k \bar {\bv}}   \right)   -   \alpha   \log   \left(\sum_{m=1}^M   \left( \frac{\bar{\bv}^{\sf H} {\bC}_{m,k} \bar {\bv} }{\bar{\bv}^{\sf H} {\bD}_{m,k} \bar {\bv} } \right)^{\beta}  \right)   \right\},
\end{align}
where $\beta = 1/(\alpha\log2)$.

Now, we propose an optimization method that solves the problem in \eqref{eq:reform_problem_apx}.
To this end, we derive the first-order optimality condition of \eqref{eq:reform_problem_apx} in the following lemma:

\begin{lemma} \label{lem:main}
The first-order optimality condition of problem \eqref{eq:reform_problem_apx} is satisfied if the following condition holds:
\begin{align} \label{eq:kkt_lem}
{\bf{B}}^{-1}_{\sf KKT} (\bar {\bv}) {\bf{A}}_{\sf KKT}(\bar {\bv}) \bar {\bv} = \lambda(\bar {\bv})  \bar {\bv},
\end{align}
where 
\begin{align} 
    \label{eq:def_lambda}
    \lambda(\bar {\bv}) 
    &= {\lambda_{\sf num}(\bar {\bv})}/{\lambda_{\sf den}(\bar {\bv})},  
    \\
    \label{eq:def_lambda_num} 
    \lambda_{\sf num}(\bar {\bv}) &= \prod_{k = 1}^{K} \left( {\bf \bar v }^{\sf H} {\bf{A}}_k \bar {\bv} \right), 
    \\
    \label{eq:def_lambda_den}
    \lambda_{\sf den}(\bar {\bv}) &= \prod_{k = 1}^{K} \left( \sum_{m = 1}^{M} \left( \frac{\bar{\bv}^{\sf H} {\bf{C}}_{m,k} \bar {\bv} }{\bar{\bv}^{\sf H} {\bf{D}}_{m,k} \bar {\bv} } \right)^{\beta}\right)^{\alpha} \left( {{\bf \bar v }^{\sf H} {\bf{B}}_k \bar {\bv}} \right),
    \\
    \label{eq:A_kkt_main}
    {\bf{A}}_{\sf KKT}(\bar {\bv}) &= \lambda_{\sf num}(\bar {\bv})   \sum_{k=1}^{K} \left( \frac{1}{\log2}  \frac{\bA_k}{{\bf \bar v }^{\sf H} \bA_k \bar {\bv}}  +   \frac{\alpha \sum_m \beta \left( \frac{\bar{\bv}^{\sf H} {\bf{C}}_{m,k} \bar {\bv} }{\bar{\bv}^{\sf H} {\bf{D}}_{m,k} \bar {\bv} }  \right)^{\beta} \left( \frac{\bD_{m,k}}{{\bf \bar v }^{\sf H} {\bD}_{m,k} \bar {\bv}} \right) }{ \sum_m  \left( \frac{\bar{\bv}^{\sf H} {\bf{C}}_{m,k} \bar {\bv} }{\bar{\bv}^{\sf H} {\bf{D}}_{m,k} \bar {\bv} }  \right)^{\beta} } \right),
    \\
    \label{eq:B_kkt_main}
    {\bf{B}}_{\sf KKT}(\bar {\bv}) &= \lambda_{\sf den}(\bar {\bv})   \sum_{k=1}^{K} \left( \frac{1}{\log2} \frac{\bB_k}{{\bf \bar v }^{\sf H} {\bB}_k \bar {\bv}}   +   \frac{\alpha \sum_m \beta \left( \frac{\bar{\bv}^{\sf H} {\bf{C}}_{m,k} \bar {\bv} }{\bar{\bv}^{\sf H} {\bf{D}}_{m,k} \bar {\bv} }  \right)^{\beta} \left( \frac{\bC_{m,k}}{{\bf \bar v }^{\sf H} \bC_{m,k} \bar {\bv}} \right) }{ \sum_m  \left( \frac{\bar{\bv}^{\sf H} {\bf{C}}_{m,k} \bar {\bv} }{\bar{\bv}^{\sf H} {\bf{D}}_{m,k} \bar {\bv} }  \right)^{\beta} } \right).
\end{align}

\end{lemma}
\begin{proof}
See Appendix \ref{proof:lem1}. 
\end{proof}


The obtained first-order optimality condition in \eqref{eq:kkt_lem} can be interpreted as a general eigenvalue problem for the matrix ${\bB}^{-1}_{\sf KKT}(\bar \bv){\bA}_{\sf KKT}(\bar \bv)$.
Accrodingly, if $\bar \bv$ is a stationary point of the problem \eqref{eq:reform_problem_apx}, it is an eigenvector of the matrix ${\bB}^{-1}_{\sf KKT}(\bar \bv){\bA}_{\sf KKT}(\bar \bv)$, where the corresponding eigenvalue is $\lambda (\bar \bv)$. 
We remark that the objective function in \eqref{eq:reform_problem_apx} is equivalant to $\log_2 \lambda(\bar \bv)$. 
This observation leads to the following fact: we need to find the leading eigenvector of ${\bB}^{-1}_{\sf KKT}(\bar \bv){\bA}_{\sf KKT}(\bar \bv)$  to maximize the objective function in \eqref{eq:reform_problem_apx}, and then such a leading eigenvector is the best stationary point among all stationary point.
We also remark that once finding the leading eigenvector, we simultaneously obtain the joint precoding and AN covariance solution with power allocation. 
The following proposition summarizes the insight:
\begin{proposition} \label{prop:optimal}
The best local optimal point for problem \eqref{eq:reform_problem_apx} is denoted by $\bar {\bv}^{\star }$, where $\bar {\bv}^{\star}$ is the leading eigenvector of ${\bf{B}}^{-1}_{\sf KKT}(\bar {\bv}^{\star}){\bf{A}}_{\sf KKT}(\bar {\bv}^{\star})$ satisfying
\begin{align}
 {\bf{B}}^{-1}_{\sf KKT}(\bar {\bv}^{\star}){\bf{A}}_{\sf KKT}(\bar {\bv}^{\star}) \bar {\bv}^{\star} = \lambda^{\star} \bar {\bv}^{\star},  
\end{align}
and $\lambda^{\star}$ is the corresponding first (largest) eigenvalue. 
\end{proposition}

\begin{algorithm} [t]
\caption{Joint and Simultaneous GPI-based Precoding (JS-GPIP)} \label{alg:perfectCSIT}
{\bf{initialize}}: $\bar {\bv}_{0}$\\
Set the iteration count $t = 1$ and GPI threshold $\epsilon_1$.\\
\While {$\left\|\bar {\bv}_{t} - \bar {\bv}_{t-1} \right\| > \epsilon_1$}
{
Create the matrices $ {\bf{A}}_{\sf KKT} (\bar {\bv}_{t-1})$ and $ {\bf{B}}_{\sf KKT} (\bar {\bv}_{t-1})$ by using \eqref{eq:A_kkt_main} and \eqref{eq:B_kkt_main}, respectively.  \\
Compute $\bar {\bv}_{t} \leftarrow \frac{{\bf{B}}_{\sf KKT}^{-1} (\bar {\bv}_{t-1}) {\bf{A}}_{\sf KKT} (\bar {\bv}_{t-1}) {\bv}_{t-1}}
{  \| {\bf{B}}_{\sf KKT}^{-1} (\bar {\bv}_{t-1}) {\bf{A}}_{\sf KKT} (\bar {\bv}_{t-1}) \bar {\bv}_{t-1} \| }$.\\
 $t \leftarrow t+1$.}
 
$\bar {\bv}^{\star} \leftarrow \bar {\bv}_{t}$\\
\Return{\ }{$\bar {\bv}^{\star}$}.
\end{algorithm}
To find the leading eigenvector, we adopt the generalized power iteration (GPI) method \cite{choi:arxiv:19} which efficiently finds the vector in an iterative manner. 
In particular, during the $t$-th iteration, we update the vector in the current iteration as
\begin{align}
    \label{eq:gpi_update}
    \bar {\bv}_{t} \leftarrow \frac{{\bf{B}}^{-1}_{\sf KKT} (\bar {\bv}_{t-1}) {\bf{A}}_{\sf KKT} (\bar {\bv}_{t-1}) \bar {\bv}_{t-1}}{ \| {\bf{B}}^{-1}_{\sf KKT} (\bar {\bv}_{t-1}) {\bf{A}}_{\sf KKT} (\bar {\bv}_{t-1}) \bar {\bv}_{t-1} \|}.
\end{align}
Until the convergence criterion $\left\|\bar {\bv}_{t} - \bar {\bv}_{t-1} \right\| < \epsilon_1$ is met, we repeat \eqref{eq:gpi_update}. 
Algorithm~\ref{alg:perfectCSIT} summarizes the GPI-based joint and simultaneous precoding and AN covariance optimization method.
\subsection{Joint and Simultaneous Optimization with Imperfect Wiretap CSIT} \label{sec:imperfect_csi}

Now, we extend the proposed algorithm to the  imperfect CSIT case of eavesdroppers in which only the  channel covariances are known for wiretap channels at the AP.
Since the perfect CSIT of the wiretap channel is not available at the AP, using the wiretap channel rate for each channel realization which is defined in \eqref{eq:se_eve} is no longer feasible.  
Consequently, we consider the wiretap channel rate as a perspective of an ergodic rate, i.e., we now use $\bbE[R_{m,k}^{\sf e}]$ instead of $R_{m,k}^{\sf e}$ as our performance metric to utilize the wiretap channel covariance.
Using the ergodic wiretap channel rate, we reformulate the optimization problem in \eqref{eq:main_problem_eve} as
\begin{align}
    \label{eq:main_problem_imperfect}
    \mathop{{\text{maximize}}}_{\bF,\bf \Phi}& \; \; \sum_{k = 1}^{K} \left[ R_k - \max_{m \in \CMcal{M}} \left\{\mathbb{E}[R_{m, k}^{\sf e}]\right\} \right]^{+}
    \\ 
     {\text{subject to}} & \;\; \sum_{i = 1}^{K} \left\| \bff_i \right\|^2 + \sum_{j = 1}^{J} \left\| \pmb {\phi}_j \right\|^2 \leq 1. 
\end{align}
To exploit the knowledge of the wiretap channel covariance, we need to properly handle $\bbE[R_{m,k}^{\sf e}]$.
To this end, we approximate $\bbE[R_{m,k}^{\sf e}]$  as
\begin{align} 
    \mathbb{E}[R_{m,k}^{\sf e}  ]   &=   \mathbb{E}   \left[   \log_2   \left(   1   +   \frac{|\bg_{m}^{\sf H} {\bf{f}}_k|^2}{\sum_{i\neq k}^{K}   |\bg_{m}^{\sf H} {\bf{f}}_i|^2   +   \sum_{j = 1}^{J}   |\bg_{m}^{\sf H} {\pmb {\phi}_j}|^2   +   \sigma_{\sf e}^2/P}   \right)   \right]
    \\ 
    &\stackrel{(a)}{\approx} \log_2   \left(   1   +   \frac{\mathbb{E}[|\bg_{m}^{\sf H} {\bf{f}}_k|^2]}{\mathbb{E}[\sum_{i\neq k}^{K}   |\bg_{m}^{\sf H} {\bf{f}}_i|^2]   +   \mathbb{E}[\sum_{j = 1}^{J}   |\bg_{m}^{\sf H} {\pmb {\phi}_j}|^2]   +   \sigma_{\sf e}^2/P}   \right)
    \\ 
    &= \log_2 \left( 1 + \frac{\bff^{\sf H}_{k} {\bR}^{\sf e}_{m} \bff_{k}}{\sum^{K}_{i \neq k} {\bff}^{\sf H}_{i} {\bR}^{\sf e}_{m} {\bff}_{i} + \sum^{J}_{j = 1} {\pmb \phi}^{\sf H}_{j} {\bR}^{\sf e}_{m} {\pmb \phi}_{j} + {\sigma}^{2}_{\sf e}/P } \right)
    \\
    \label{eq:Rmk_approx}
    &= \tilde R_{m, k}^{\sf e},
\end{align}
where $(a)$ follows from Lemma 1 in \cite{6816003}.
Replacing $\mathbb{E}[R_{m,k}^{\sf e}  ]$ in \eqref{eq:main_problem_imperfect} with $\tilde R_{m,k}^{\sf e}$ in \eqref{eq:Rmk_approx}, the problem in \eqref{eq:main_problem_imperfect}  becomes
\begin{align}
    \label{eq:obj_func_ub}
     \mathop{{\text{maximize}}}_{\bF,\bf \Phi} & \;\; \sum_{k = 1}^{K} \left[ R_k - \max_{m \in \CMcal{M}} \left\{\tilde R_{m, k}^{\sf e}\right\} \right]^{+}
    \\
     {\text{subject to}} & \;\; \sum_{i = 1}^{K} \left\| \bff_i \right\|^2 + \sum_{j = 1}^{J} \left\| \pmb {\phi}_j \right\|^2 \leq 1.
\end{align}

 
Now, following the similar steps as in Section \ref{sec:Opt_perfectCSIT}, we can derive the first-order optimality condition of the problem in \eqref{eq:obj_func_ub}.
Let us define $\bar {{\bC}}_{m,k}$ and $\bar {{\bD}}_{m,k}$ as
\begin{align} 
    \label{eq:C_matrix_const_eve2}
    &\bar {{\bC}}_{m,k} = {\rm blkdiag}   \left(\bR_{m}^{\sf e}, ..., \bR_{m}^{\sf e}  \right) + {\bf{I}}_{N   (   K+J   )} \frac{ \sigma_{\sf e}^2}{P},   
    \\
    \label{eq:D_matrix_const_eve2} 
    & \bar {{\bD}}_{m,k} = \bar {{\bC}}_{m,k} - {\rm blkdiag} ({\bf{0}}, \cdots,   \underbrace{\bR_{m}^{\sf e}}_{k{\rm th} \; {\rm block}}  , \cdots, {\bf{0}} ).
\end{align}
Then, the first order optimality condition in the form of the generalized eigenvalue problem is
\begin{lemma} \label{lem:imperfectCSIT}
    The first-order optimality condition of problem \eqref{eq:obj_func_ub} is satisfied if the following condition holds:
    \begin{align} \label{eq:kkt_secB}
        {\bar \bB}^{-1}_{\sf KKT} (\bar {\bv}) {\bar \bA}_{\sf KKT}(\bar {\bv}) \bar {\bv} = \bar{\lambda}(\bar {\bv})  \bar {\bv},
    \end{align}
    where
    \begin{align} 
        \label{eq:def_lambda_bar}
        \bar \lambda(\bar {\bv}) 
        &= {{\lambda}_{\sf num}(\bar {\bv})}/{\bar {\lambda}_{\sf den}(\bar {\bv})},
    \\
     \label{eq:def_lambda_den_bar}
    \bar {\lambda}_{\sf den}(\bar {\bv}) &= \prod_{k = 1}^{K} \left( \sum_{m = 1}^{M} \left( \frac{\bar{\bv}^{\sf H} \bar {\bf{C}}_{m,k} \bar {\bv} }{\bar{\bv}^{\sf H} \bar {\bf{D}}_{m,k} \bar {\bv} } \right)^{\beta}\right)^{\alpha} \left( {{\bf \bar v }^{\sf H} {\bf{B}}_k \bar {\bv}} \right),
    \\
    \label{eq:A_kkt_imperfect}
    \bar {\bf{A}}_{\sf KKT}(\bar {\bv}) &= {\lambda}_{\sf num}(\bar {\bv}) \sum_{k=1}^{K} \left( \frac{1}{\log2} \frac{\bA_k}{{\bf \bar v }^{\sf H} \bA_k \bar {\bv}} +  \frac{\alpha \sum_m \beta \left( \frac{\bar{\bv}^{\sf H} {\bar {\bf{C}}}_{m,k} \bar {\bv} }{\bar{\bv}^{\sf H} \bar {{\bf{D}}}_{m,k} \bar {\bv} }  \right)^{\beta} \left( \frac{\bar {\bD}_{m,k}}{{\bf \bar v }^{\sf H} {\bar {\bD}}_{m,k} \bar {\bv}} \right) }{ \sum_m  \left( \frac{\bar{\bv}^{\sf H} \bar {{\bf{C}}}_{m,k} \bar {\bv} }{\bar{\bv}^{\sf H} \bar {{\bf{D}}}_{m,k} \bar {\bv} }  \right)^{\beta} } \right),
    \\
    \label{eq:B_kkt_imperfect}
    \bar {\bf{B}}_{\sf KKT}(\bar {\bv}) &= \bar {\lambda}_{\sf den}(\bar {\bv}) \sum_{k=1}^{K} \left( \frac{1}{\log2} \frac{\bB_k}{{\bf \bar v }^{\sf H} {\bB}_k \bar {\bv}} + \frac{\alpha \sum_m \beta \left( \frac{\bar{\bv}^{\sf H} {\bar {\bf{C}}}_{m,k} \bar {\bv} }{\bar{\bv}^{\sf H} \bar {{\bf{D}}}_{m,k} \bar {\bv} }  \right)^{\beta} \left( \frac{\bar {\bC}_{m,k}}{{\bf \bar v }^{\sf H} {\bar {\bC}}_{m,k} \bar {\bv}} \right) }{ \sum_m  \left( \frac{\bar{\bv}^{\sf H} \bar {{\bf{C}}}_{m,k} \bar {\bv} }{\bar{\bv}^{\sf H} \bar {{\bf{D}}}_{m,k} \bar {\bv} }  \right)^{\beta} } \right),
    \end{align}
    and $\bA_k$, $\bB_k$, and $\lambda_{\sf num}(\bar {\bv})$ are defined in \eqref{eq:A_matrix_const_user}, \eqref{eq:B_matrix_const_user}, and \eqref{eq:def_lambda_num}, respectively.
\end{lemma}
Algorithm~\ref{alg:perfectCSIT} can directly find the leading eigenvector of \eqref {eq:kkt_secB} by replacing  ${\bA}_{\sf KKT}$ and ${\bB}_{\sf KKT}$ with $\bar {\bA}_{\sf KKT}$ and $\bar {\bB}_{\sf KKT}$, respectively. 
We call the algorithm as JS-GPIP (Cov) throughout this paper.
Consequently, we can jointly and simultaneously design the secure precoder, AN covariance matrix, and precoder-AN power allocation when only the channel covariance is available for the wiretap channel.

\begin{remark} \label{rm:S-GPIP}
\normalfont (Secure GPI-based precoding)
    The proposed JS-GPIP and JS-GPIP (Cov) algorithms naturally reduce to a secure GPI-based precoding algorithm (S-GPIP) without AN for the perfect wiretap CSIT and imperfect wiretap CSIT cases by enforcing $\bf \Phi = {\bf 0}$, i.e., $J = 0$, respectively. 
\end{remark}

\subsection{Application to Null-Space Projection Approach}

\subsubsection{Alternating Approach}
Although using null-space projection of AN covariance is not optimal, its design principle is greatly intuitive.
In this regard, designing the AN covariance based on the null-space of the users' channels has been a popular suboptimal scheme for maximizing the sum secrecy rate \cite{1558439,7328729}.
In this regard, we also show the application of our framework to the null-space projection approach.
We design $\bf \Phi$ as the null-space of user channels, and $\bF$ as the proposed GPI-based secure precoder.
Then we further need to perform power allocation between  $\bf \Phi$ and $\bF$ in an alternating manner.
To this end, we employ a line search method by explicitly defining a power fraction factor $0\leq \xi \leq 1$.

To begin with, we solve the optimization problem in \eqref{eq:main_problem_eve} for the precoder with the given power fraction factor $\xi$ and AN covariance matrix that is designed as the null-space of the user channels.
Let $\bf \Phi = \sqrt{\xi} \tilde{\bf \Phi}$ and $\bF = \sqrt{1-\xi} \tilde{\bF}$,
where $\tilde{\bf \Phi} = [\bI_{\mathit{N}}-(\bH(\bH^{\mathsf{H}}\bH)^{\mathrm {-1}})^{\mathsf{H}}]_{:,1:J}$ where $[\cdot]_{:,1:J}$ indicates using the first $J$ columns.
Since we assume using maximum transmit power, we have the transmit power to be
\begin{align}
    \label{eq:trace}
    &{\sf{tr}} ((1-\xi) \tilde{\bF} \tilde{\bF}^{\sf H}) + {\sf{tr}} (\xi \tilde{\bf \Phi} \tilde{\bf \Phi}^{\sf H}) = 1.
\end{align}
Equivalently, we rewrite \eqref{eq:trace} as
\begin{align}
    \label{eq:trace2}
    &{\sf{tr}} \left( \frac{(1-\xi)}{1 - {\sf{tr}} (\xi \tilde{\bf \Phi} \tilde{\bf \Phi}^{\sf H})} \tilde{\bF} \tilde{\bF}^{\sf H} \right)= 1.
\end{align}
Now, we define $\bW$ as
\begin{align}
    \label{eq:W_defined}
    &\bW 
    = \frac{1}{\sqrt{1 - {\sf{tr}} (\xi \tilde{\bf \Phi} \tilde{\bf \Phi}^{\sf H})}} \bF.
\end{align}
From \eqref{eq:trace2} and \eqref{eq:W_defined}, we have the power constraint to be ${\sf{tr}} (\bW \bW^{\sf H}) = 1$. 
Leveraging \eqref{eq:W_defined}, we rewrite $R_k$ and $R_{m,k}^{\sf e}$ as
\begin{align} \label{eq:se_user_NS}
    &R_{k}(\bW)= \log_2 \left( 1 + \frac{(1 - {\sf{tr}} (\xi \tilde{\bf \Phi} \tilde{\bf \Phi}^{\sf H}))|\bh_{k}^{\sf H} \bw_k|^2}{( 1 -  {\sf{tr}} (\xi \tilde{\bf \Phi} \tilde{\bf \Phi}^{\sf H}) \sum_{i = 1, i\neq k}^{K} |\bh_{k}^{\sf H} \bw_i|^2 + \xi \sum_{j = 1}^{J}   |{\bf{h}}_{k}^{\sf H} {\tilde{\pmb {\phi}_j}}|^2  + \sigma^2/P} \right)\ \ \ 
\end{align}
and
\begin{align}
    \label{eq:se_eve_NS}
    &R_{m,k}^{\sf e}(\bW) = \log_2   \left(1 + \frac{(1 - {\sf{tr}} (\xi \tilde{\bf \Phi} \tilde{\bf \Phi}^{\sf H}))|\bg_{m}^{\sf H} \bw_k|^2}{(1  - {\sf{tr}} (\xi \tilde{\bf \Phi} \tilde{\bf \Phi}^{\sf H}))   \sum_{i = 1, i\neq k}^{K} |\bg_{m}^{\sf H} \bw_i|^2 + \xi \sum_{j = 1}^{J}   |{\bg_{m}^{\sf H} {\tilde{\pmb {\phi}_j}}|^2 + \sigma_{\sf e}^2/P}} \right).
\end{align}

Similar to \eqref{eq:max_rate}, we approximate the non-smooth maximum objective function by applying the LogSumExp function  with $\bar \bw = \mathrm{vec}(\bW)$.
Then the problem \eqref{eq:main_problem_eve} for given $\bf \Phi$ is reformulated as 
\begin{align} \label{eq:approx_problem_eve_hat}
    \mathop{{\text{maximize}}}_{\bar\bw}&   \sum_{k = 1}^{K}   \left\{   \log_2   \left(   \frac{{{\bar \bw} }^{\sf H} { {\bA}^{\sf ns}_k {\bar \bw}} }{{{\bar \bw} }^{\sf H} {\bB}^{\sf ns}_k  {\bar \bw}}   \right)   -   \alpha   \log   \left(\sum_{m=1}^M   \left(   \frac{{\bar \bw}^{\sf H}  {\bC}^{\sf ns}_{m,k} {\bar \bw} }{\bar{\bw}^{\sf H}  {\bD}^{\sf ns}_{m,k} {\bar \bw} }   \right)^{\beta}  \right)   \right\},
\end{align}
where
\begin{align}
    \label{eq:A_matrix_const_eve_hat_HC}
   {\bA}^{\sf ns}_k &= (1 - {\sf{tr}} (\xi \tilde{\bf \Phi} \tilde{\bf \Phi}^{\sf H})){\rm blkdiag}\left({\bf{h}}_k {\bf{h}}_k^{\sf H}, ..., {\bf{h}}_k {\bf{h}}_k^{\sf H}  \right) + {\bf{I}}_{NK} (\xi   \sum_{j = 1}^{N}  |{\bf{h}}_{k}^{\sf H} {\tilde{\pmb {\phi}_j}}|^2  +  \frac{\sigma^2}{P}), 
    \\
    \label{eq:B_matrix_const_eve_hat_HC}
    {\bB}^{\sf ns}_k &= {\bA}^{\sf ns}_k - (1 - {\sf{tr}} (\xi \tilde{\bf \Phi} \tilde{\bf \Phi}^{\sf H})){\rm blkdiag}({\bf{0}}, \cdots,  \underbrace{{\bh}_k {\bh}_k^{\sf H}}_{k{\rm th}\; {\rm block}} , \cdots, {\bf{0}} ),
    \\
    \label{eq:C_matrix_const_eve_hat_HC}
    {\bC}^{\sf ns}_{m,k} &= (1 - {\sf{tr}} (\xi \tilde{\bf \Phi} \tilde{\bf \Phi}^{\sf H})){\rm blkdiag}\left({\bf{g}}_m {\bf{g}}_m^{\sf H}, ..., {\bf{g}}_m {\bf{g}}_m^{\sf H}  \right) + {\bf{I}}_{NK} (\xi   \sum_{j = 1}^{N}   |{\bf{g}}_{m}^{\sf H} {\tilde{\pmb {\phi}_j}}|^2  +  \frac{\sigma_{\sf e}^2}{P}) ,    
    \\
    \label{eq:D_matrix_const_eve_hat_HC} 
    {\bD}^{\sf ns}_{m,k}  &= {\bC}^{\sf ns}_{m,k}  - (1 - {\sf{tr}} (\xi \tilde{\bf \Phi} \tilde{\bf \Phi}^{\sf H}) ){\rm blkdiag}({\bf{0}},  \cdots , \underbrace{{\bg}_m {\bg}_m^{\sf H}}_{k{\rm th}\; {\rm block}} ,  \cdots, {\bf{0}} ). 
\end{align}
We also derive the first-order optimality condition of \eqref{eq:approx_problem_eve_hat} in the following lemma:
\begin{lemma} \label{lem:KKT2}
    The first-order optimality condition of problem \eqref{eq:approx_problem_eve_hat} is satisfied if the following condition holds: 
    \begin{align} \label{eq:kkt_lem2}
     {\bf{B}}^{{\sf ns} -1}_{\sf KKT} (\bar \bw) {\bf{A}}^{\sf ns}_{\sf KKT}(\bar \bw) \bar \bw = \lambda^{\sf ns}(\bar \bw) \bar \bw,
    \end{align}
where
\begin{align}
    \label{eq:def_lambda_ns}
    \lambda^{\sf ns}(\bar {\bw}) 
    &= {{\lambda}^{\sf ns}_{\sf num}(\bar {\bw})}/{{\lambda}^{\sf ns}_{\sf den}(\bar {\bw})},
    \\
    {\lambda}^{\sf ns}_{\sf num}(\bar {\bw}) &= \prod_{k = 1}^{K} \left( {\bar \bw }^{\sf H} {\bf{A}}^{\sf ns}_k \bar {\bw} \right), \label{eq:def_lambda_num_hat} 
    \\
    \label{eq:def_lambda_den_ns}
   {\lambda}^{\sf ns}_{\sf den}(\bar {\bw}) &= \prod_{k = 1}^{K} \left( \sum_{m = 1}^{M} \left( \frac{\bar{\bw}^{\sf H} \bC^{\sf ns}_{m,k} \bar {\bw} }{\bar{\bw}^{\sf H}  \bD^{\sf ns}_{m,k} \bar {\bw} } \right)^{\beta}\right)^{\alpha} \left( {\bar \bw }^{\sf H} \bB^{\sf ns}_k \bar {\bw} \right).
   \\
    \label{eq:A_kkt_ns}
    {\bf{A}}^{\sf ns}_{\sf KKT}(\bar \bw) &= {\lambda}^{\sf ns}_{\sf num}(\bar \bw)   \sum_{k=1}^{K} \left( \frac{1}{\log2}  \frac{\bA^{\sf ns}_k}{{\bar \bw }^{\sf H} \bA^{\sf ns}_k \bar {\bw}} + \frac{\alpha \sum_m \beta \left( \frac{\bar{\bw}^{\sf H}  {\bC}^{\sf ns}_{m,k} \bar {\bw} }{\bar{\bw}^{\sf H} {\bD}^{\sf ns}_{m,k} \bar {\bw} }  \right)^{\beta} \left( \frac{{\bD}^{\sf ns}_{m,k}}{{\bar \bw }^{\sf H} {\bD}^{\sf ns}_{m,k} \bar {\bw}} \right) }{ \sum_m  \left( \frac{\bar{\bw}^{\sf H} {\bC}^{\sf ns}_{m,k} \bar {\bw} }{\bar{\bw}^{\sf H} {\bD}^{\sf ns}_{m,k} \bar {\bw} }  \right)^{\beta} } \right),
    \\
    \label{eq:B_kkt_ns}
    {\bB}^{\sf ns}_{\sf KKT}(\bar \bw) &= {\lambda}^{\sf ns}_{\sf den}(\bar {\bw})  \sum_{k=1}^{K} \left( \frac{1}{\log2}  \frac{\bB^{\sf ns}_k}{{\bar \bw }^{\sf H} {\bB}^{\sf ns}_k \bar {\bw}} + \frac{\alpha \sum_m \beta \left( \frac{\bar{\bw}^{\sf H} {\bC}^{\sf ns}_{m,k} \bar {\bw} }{\bar{\bw}^{\sf H} {\bD}^{\sf ns}_{m,k} \bar {\bw} }  \right)^{\beta} \left( \frac{{\bC}^{\sf ns}_{m,k}}{{\bar \bw }^{\sf H} {\bC}^{\sf ns}_{m,k} \bar {\bw}} \right) }{ \sum_m  \left( \frac{\bar{\bw}^{\sf H} {\bC}^{\sf ns}_{m,k} \bar {\bw} }{\bar{\bw}^{\sf H} {\bD}^{\sf ns}_{m,k} \bar {\bw} }  \right)^{\beta} } \right).
\end{align}
\end{lemma}
\begin{algorithm} [t]
\caption{Joint GPI-based Precoding with Null-Space AN (J-GPIP-NS)} \label{alg:J-GPIP-NS}
{\bf{initialize}}:  $\bar \bff^{(0)} =  {\rm vec}({\bF}^{(0)})$ and $\tilde{\bf{\Phi}}=[\bI_{\mathit{N}}-(\bH(\bH^{\mathsf{H}}\bH)^{\mathrm {-1}})^{\mathsf{H}}]_{:,1:J}$.\\
Set  iteration counts $t = 1$ and $i=1$, $\xi^{(0)} = 0$, step size of power fraction factor $\Delta \xi$, and GPI threshold $\epsilon_2$.\\
\While {$ \xi^{(i)} \leq 1$}
{
    Compute ${\bf\Phi}^{(i)} = \sqrt{\xi^{({i})}} \tilde{\bf \Phi}$ and $\bar\bw_\mathrm{0}^{(i)} = \frac{1}{\sqrt{1 - {\sf{tr}} (\xi^{(\mathrm{i})} \tilde{\bf \Phi} \tilde{\bf \Phi}^{\sf H})}}\bar \bff^{(i)}$.\\
    \While {$\left\|\bar {\bw}^{(i)}_{t} - \bar {\bw}^{(i)}_{t-1} \right\| > \epsilon_2$}
    {
    Create the matrices $ {\bf{A}}^{\sf ns}_{\sf KKT} (\bar{\bw}_{t-1}^{(i)})$ and $ {\bB}^{\sf ns}_{\sf KKT} (\bar{\bw}_{t-1}^{(i)})$ by using \eqref{eq:A_kkt_ns}, \eqref{eq:B_kkt_ns}, and $\xi^{(i)}$.  \\
    Compute $\bar{\bw}_{t}^{(i)} \leftarrow \frac{ {\bB}^{{\sf ns} -1}_{\sf KKT} (\bar {\bw}_{t-1}^{(i)}) {\bA}^{\sf ns}_{\sf KKT} (\bar {\bw}_{t-1}^{(i)}) \bar {\bw}_{t-1}^{(i)} }{ \| {\bB}^{{\sf ns} -1}_{\sf KKT} (\bar {\bw}_{t-1}^{(i)}) {\bA}^{\sf ns}_{\sf KKT} (\bar {\bw}_{t-1}^{(i)}) \bar {\bw}_{t-1}^{(i)} \| }$.\\
     $t \leftarrow t+1$.}
     $\bar\bff^{(i)} \leftarrow \bar \bw_t^{(i)} \sqrt{1 - {\sf{tr}} (\xi^{(i)} \tilde{\bf \Phi} \tilde{\bf \Phi}^{\sf H})}$. \\
     $\xi^{(i+1)} \leftarrow \xi^{(i)} + \Delta \xi$. \\
     $i\leftarrow i+1.$
 }
$[\bar\bff^{\star}, \ {\bf \Phi}^{\star}] = \argmax_{(\bar \bff^{(i)},{\bf \Phi}^{(i)})} \sum_{\mathit{k} \mathrm{ = 1}}^{\mathit{K}}\left[\mathit{R_k}-\max_{\mathit{m} \in \CMcal{M}}\left\{\mathit{R}_{\mathit{m, k}}^{\sf e}\right\} \right]^{+}$. \\

\Return{\ }{$\bF^{\star}$ and ${\bf {\Phi}}^{\star}$}.
\end{algorithm}
Subsequently, we can find the leading eigenvector  of \eqref{eq:kkt_lem2} by using the GPI method, and derive the precoder $\bF$ based on \eqref{eq:W_defined} accordingly.
We perform the same procedure by increasing the power fraction factor $\xi$ to derive $\bF$ for different power allocation cases.
Then we select the power fraction factor $\xi^\star$ that offers the highest sum secrecy rate and use the corresponding pair of $\bF^\star $ and $\bf \Phi^\star$.
We summarize the proposed joint  GPI-based precoding with null-space AN approach (J-GPIP-NS) in Algorithm~\ref{alg:J-GPIP-NS}.

\subsubsection{Low-complexity (non-alternating) Approach} 
We note that J-GPIP-NS repeatedly computes GPI for each $\xi^{(i)}$, which causes a high computational burden.
To reduce computational complexity of J-GPIP-NS, we propose a low complexity version of the null-space projection based GPI method that computes the precoder by ignoring the AN and also compute the AN covariance matrix as the null-space projection matrix. 
Then, by defining a power fraction factor $\xi_L$, we perform line searching and select the precoder and AN pair with the highest performance.
In this case, only the power scaling of the precoder and AN covariance changes over line searching.
This approach is described in Algorithm~\ref{alg:J-GPIP-NS-Low}.
Let $\bar \bff = {\rm vec}(\bF)$. 
Then $\hat \bA^{\sf ns}_{\sf KKT}$ and $\hat \bB^{\sf ns}_{\sf KKT}$ can be  computed by setting $\xi = 0$ in  \eqref{eq:A_kkt_ns} and \eqref{eq:B_kkt_ns} as follows:
\begin{align} \label{eq:A_kkt_hat_ns}
    &\hat {\bf{A}}^{\sf ns}_{\sf KKT}(\bar \bff) = \hat {\lambda}^{\sf ns}_{\sf num}(\bar \bff) \sum_{k=1}^{K} \left( \frac{1}{\log2}  \frac{\hat \bA^{\sf ns}_k}{{\bar \bff }^{\sf H} \hat  \bA^{\sf ns}_k \bar {\bff}} + \frac{\alpha \sum_m \beta \left( \frac{\bar{\bff}^{\sf H} \hat {\bC}^{\sf ns}_{m,k} \bar {\bff} }{\bar{\bff}^{\sf H} \hat {\bD}^{\sf ns}_{m,k} \bar {\bff} }  \right)^{\beta} \left( \frac{\hat {\bD}^{\sf ns}_{m,k}}{{\bar \bff }^{\sf H} \hat {\bD}^{\sf ns}_{m,k} \bar {\bff}} \right) }{ \sum_m  \left( \frac{\bar{\bff}^{\sf H} \hat {\bC}^{\sf ns}_{m,k} \bar {\bff} }{\bar{\bff}^{\sf H} \hat {\bD}^{\sf ns}_{m,k} \bar {\bff} }  \right)^{\beta} } \right),
\end{align}
\begin{align} \label{eq:B_kkt_hat_ns}
    \hat {\bB}^{\sf ns}_{\sf KKT}(\bar \bff) &= \hat {\lambda}^{\sf ns}_{\sf den}(\bar {\bff}) \sum_{k=1}^{K} \left( \frac{1}{\log2}  \frac{\hat \bB^{\sf ns}_k}{{\bar \bff }^{\sf H} \hat {\bB}^{\sf ns}_k \bar {\bff}} + \frac{\alpha \sum_m \beta \left( \frac{\bar{\bff}^{\sf H} \hat {\bC}^{\sf ns}_{m,k} \bar {\bff} }{\bar{\bff}^{\sf H} \hat {\bD}^{\sf ns}_{m,k} \bar {\bff} }  \right)^{\beta} \left( \frac{\hat {\bC}^{\sf ns}_{m,k}}{{\bar \bff }^{\sf H} \hat {\bC}^{\sf ns}_{m,k} \bar {\bff}} \right) }{ \sum_m  \left( \frac{\bar{\bff}^{\sf H} \hat {\bC}^{\sf ns}_{m,k} \bar {\bff} }{\bar{\bff}^{\sf H} \hat {\bD}^{\sf ns}_{m,k} \bar {\bff} }  \right)^{\beta} } \right),
\end{align}
where
\begin{align}
\label{eq:def_lambda_hat_ns}
    \hat \lambda^{\sf ns}(\bar {\bff}) 
    &= {\hat {\lambda}^{\sf ns}_{\sf num}(\bar {\bff})}/{\hat {\lambda}^{\sf ns}_{\sf den}(\bar {\bff})},
    \\
    \hat {\lambda}^{\sf ns}_{\sf num}(\bar {\bff}) &= \prod_{k = 1}^{K} \left( {\bar \bff }^{\sf H} \hat {\bA}^{\sf ns}_k \bar {\bff} \right), \label{eq:def_lambda_num_check} 
    \\
      \label{eq:def_lambda_den_hat_ns}
   \hat {\lambda}^{\sf ns}_{\sf den}(\bar {\bff}) &= \prod_{k = 1}^{K} \left( \sum_{m = 1}^{M} \left( \frac{\bar{\bff}^{\sf H} \hat {\bC}^{\sf ns}_{m,k} \bar {\bff} }{\bar{\bff}^{\sf H} \hat {\bD}^{\sf ns}_{m,k} \bar {\bff} } \right)^{\beta}\right)^{\alpha} \left( {{\bar \bff }^{\sf H} \hat {\bB}^{\sf ns}_k \bar {\bff}} \right).
\end{align}
\begin{algorithm} [t]
\caption{J-GPIP-NS with Low Complexity} \label{alg:J-GPIP-NS-Low} 
{\bf{initialize}}: $\bar \bff_{0} =  {\rm vec}({\bF}_0)$.\\
Set iteration counts $t = 1$ and $i=1$, $\xi_{L}^{(0)}=0$,  $\Delta \xi_{L}$, and $\epsilon_3$.\\
\While {$\left\|\bar {\bff}_{t} - \bar {\bff}_{t-1} \right\| > \epsilon_3$}
{
Create the matrices $ \hat {\bA}^{\sf ns}_{\sf KKT} (\bar {\bff}_{t-1})$ and $\hat {\bB}^{\sf ns}_{\sf KKT} (\bar {\bff}_{t-1})$ by using \eqref{eq:A_kkt_hat_ns} and \eqref{eq:B_kkt_hat_ns}.   \\
Compute $\bar {\bff}_{t} \leftarrow \frac{\hat {\bB}^{{\sf ns} -1}_{\sf KKT} (\bar {\bff}_{t-1}) \hat {\bA}^{\sf ns}_{\sf KKT} (\bar {\bff}_{t-1}) \bar {\bff}_{t-1} }{ \| \hat {\bB}^{{\sf ns} -1}_{\sf KKT} (\bar {\bff}_{t-1}) \hat {\bA}^{\sf ns}_{\sf KKT} (\bar {\bff}_{t-1}) \bar {\bff}_{t-1} \| }$.\\
 $t \leftarrow t+1$.}
 Compute $\tilde{\bf\Phi}=[\bI_{\mathit{N}}-(\bH(\bH^{\mathsf{H}}\bH)^{\mathrm {-1}})^{\mathsf{H}}]_{:,1:J}$\\
\While {$\xi_L^{(i)} \leq 1$}
{
      $\bar {\bff}^{(i)} \leftarrow \sqrt{1-\xi_{L}^{(i)}}\; \bar {\bff}_{t}$. \\
      $  {\bf \Phi}^{(i)} \leftarrow \sqrt{\xi_{L}^{(i)}} \tilde {\bf \Phi}$. \\
     $\xi_{L}^{(i+1)} \leftarrow \xi_{L}^{(i)} + \Delta\xi_L$. \\
     $i\leftarrow i+1.$
     \\
 }
 
$[\bar {\bff}^{\star}, {\bf {\Phi}}^{\star}] = \argmax_{(\bar \bff^{(i)}, {\bf\Phi}^{(i)})} \sum_{\mathit{k} \mathrm{ = 1}}^{\mathit{K}}\left[\mathit{R_k}-\max_{\mathit{m} \in \CMcal{M}}\left\{\mathit{R}_{\mathit{m, k}}^{\sf e}\right\} \right]^{+}$. \\

\Return{\ }{${\bF}^{\star}$ and $ {\bf {\Phi}}^{\star}$}.
\end{algorithm}

\subsection{Complexity Analysis} \label{subsec:Alg_comp}
The computational complexity of Algorithm \ref{alg:perfectCSIT} is dominated by the matrix inversion of ${\bB}_{\sf KKT} (\bar {\bv})$.
Generally, with Gauss–Jordan elimination, computational complexity for the inversion of a $N(K+J) \times N(K+J)$ is order of $\CMcal{O} \left((N(K+J))^3\right)$.
In our case, we can exploit the structure of ${\bB}_{\sf KKT} (\bar {\bv})$ to reduce the complexity.
In particular, since ${\bB}_{\sf KKT} (\bar {\bv})$ is constructed based on a linear combination of block-diagonal and Hermitian matrices ${\bB}_k$ and $\bC_{m,k}$, ${\bB}_{\sf KKT} (\bar {\bv})$ is also a block-diagonal and Hermitian matrix.
Therefore, computing ${\bB}^{-1}_{\sf KKT} (\bar {\bv}^{})$ requires a total of $\CMcal{O} \left( \frac{1}{3}\max (K,J) N^3 \right)$ complexity \cite{choi:arxiv:19}.
Assuming that we have $T_1$ total iterations,  the total complexity of Algorithm \ref{alg:perfectCSIT} is  $\CMcal{O} \left( \frac{1}{3} T_1 \max (K,J) N^3 \right)$.
Similarly, Algorithm \ref{alg:J-GPIP-NS} has a complexity of $\CMcal{O}\left( \frac{1}{3}T_2 K N^3 L \right)$ when we assume that the power iteration per line search takes the similar number of iterations denoted as $T_2$ and consider $L$ to be the number of total line search with $NK \times NK$ block-diagonal Hermitian matrix $\bB^{\sf ns}_{\sf KKT} (\bar {\bw})$.
With an assumption of $T_3$ iterations for the power iteration, Algorithm \ref{alg:J-GPIP-NS-Low} has a total of $\CMcal{O}\left(\frac{1}{3}T_3 K N^3 \right)$ complexity since the line searching method of Algorithm \ref{alg:J-GPIP-NS-Low} does not increase the order of computational complexity.
In particular, for $J \le K$ and $T_1\approx T_2 \approx T_3 = T$, JS-GPIP and J-GPIP-NS (Low) have the lowest complexity of $\CMcal{O}\left( \frac{1}{3}TKN^3 \right)$. 
Such a complexity order is same as a representative low-complexity framework in sum rate maximization (not secrecy rate), namely, the weighted minimum mean square error (MMSE) method \cite{chris:twc:08}.
In addition, the complexity order is also substantially small compared to the existing secure precoding methods. 
For instance, a two-stage optimization approach including solving a sequence of semi-definite programs (SDPs) was proposed with computational complexity order of $\CMcal{O}\left(N^{6.5}\right)$  for a single confidential message and multicast message \cite{mei2017artificial}. 
For visible light communications, a SDP-based algorithm was also developed for a single legitimate user and multiple eavesdroppers with computational complexity order of $O\left(N^{8.5}\right)$ \cite{shi2020artificial}.
As noted, the complexity of JS-GPIP is significantly lower than the state-of-the-art secure precoding algorithms.

\section{Numerical Results} \label{sec:nume}


In this section, we evaluate the performance of the proposed algorithms under our secure precoding framework.
In simulations, we employ the one-ring channel model in \cite{6542746} for generating the small scale fading effects of the channels. 
Regarding the path-loss, we adopt the ITU-R model in \cite{6586668} which models indoor non-line-of-sight path-loss environments.
We use a distance power-loss coefficient as $30$ (equivalent to path-loss exponent of $3$), carrier frequency of $5$ GHz, $10$ MHz bandwidth (passband), $10 - 12$ dB lognormal shadowing variance, and $5$ dB noise figure are considered.
We assume the noise power spectral density of legitimate users and eavesdroppers are the same as $-174$ dBm/Hz.
Users are randomly generated around the AP with the maximum distance of $50\;m$ and minimum distance of $5\; m$ from the AP.
Eavesdroppers are randomly distributed around random users with the maximum distance of $5\; m$ from the users to overhear the transmitted signal.

 \begin{figure}[t] 
    \begin{center}
    \includegraphics[width=.6\linewidth]{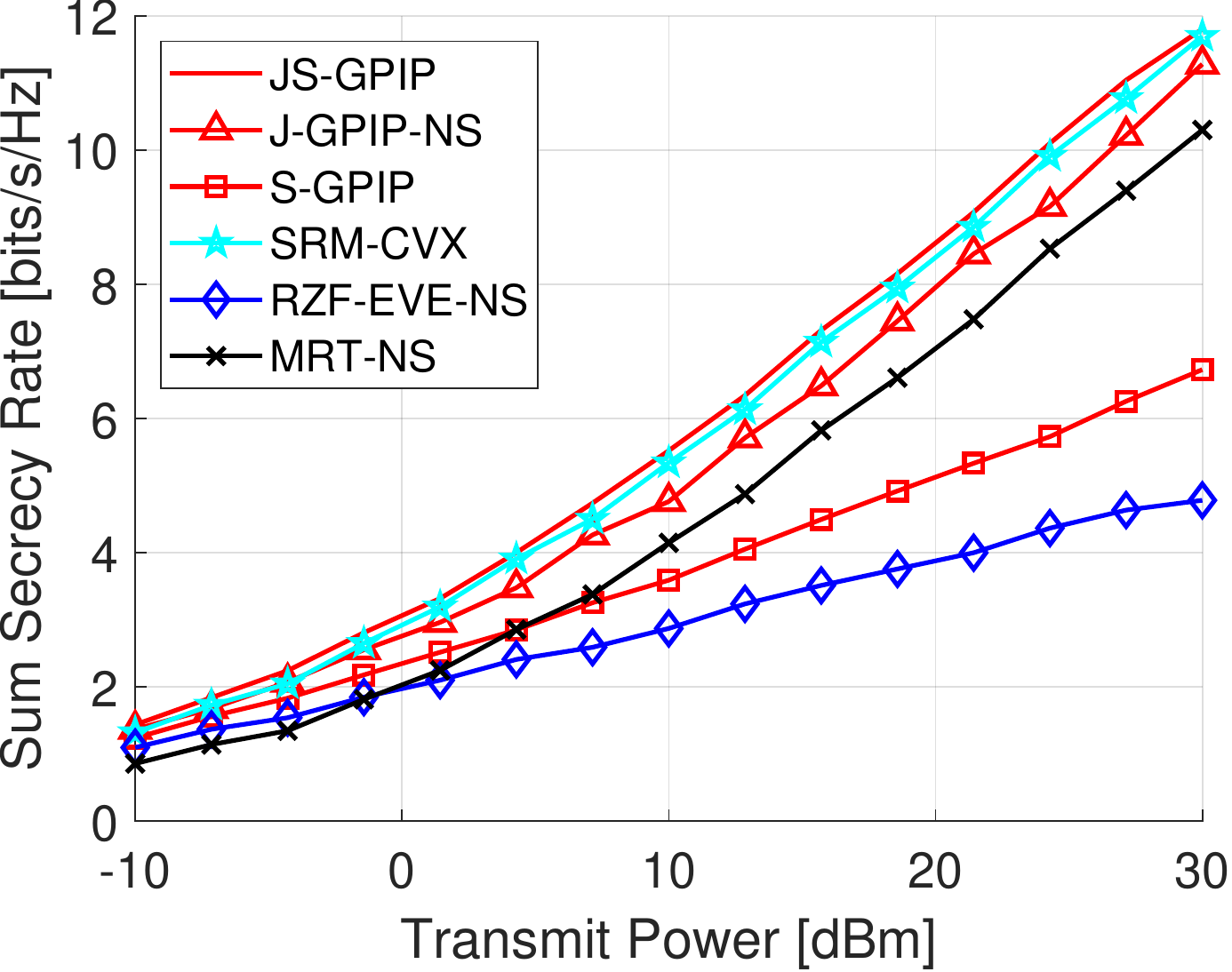}
    \end{center}
    \caption{Sum secrecy rate results for transmit power $P$ with $N = 4$ AP antennas, $K = 1$ users, and $M = 3$ eavesdroppers.}
    \label{fig:SRM_comp}
\end{figure}


Including the proposed algorithms and benchmarks, we evaluate the following cases:
\begin{itemize}
    \item The proposed algorithms: JS-GPIP, JS-GPIP (Cov), J-GPIP-NS, J-GPIP-NS (Low), and S-GPIP.
Recall that JS-GPIP (Cov) is the proposed JS-GPIP for the imperfect wiretap CSIT case and S-GPIP is the proposed JS-GPIP with $\bf \Phi = 0$ as discussed in Remark~\ref{rm:S-GPIP}.
    \item Benchmarks: $(i)$ GPIP in \cite{choi:arxiv:19},  $(ii)$ regularized zero-forcing (RZF), $(iii)$ RZF with considering eavesdroppers in precoding design (RZF-EVE), and $(iv)$ the applications of the benchmarks to the null-space AN projection approach: GPIP-NS, RZF-NS, RZF-EVE-NS, and maximum ratio combining with  null-space AN (MRT-NS).
    \item A CVX-based method (single user case): CVX-based secrecy rate maximization (SRM-CVX) in \cite{6482662} which considers a single user with multiple eavesdroppers.
\end{itemize}
In particular, RZF-EVE-NS builds the $N \times K$ RZF precoder by including $K$ users and selecting $M_{\max}$ eavesdroppers with the highest channel gains $||\bg_{m}||$ where $M_{\max}=\min(N-K,M)$.
The conventional linear precoders with NS extension adopt the line searching method to find the power ratio between precoder and null-space AN, which is similar to Algorithm \ref{alg:J-GPIP-NS-Low}.
We set $J = N$ throughout the simulations unless mentioned otherwise.
We initialize the precoder of the proposed methods with ZF and the AN covariance matrix with the user channel null-space matrix.
In Fig.~\ref{fig:SRM_comp}, we consider $N=4$ AP antennas, $K=1$ users, and $M=3$ eavesdroppers and compare the sum secrecy rates of the simulated algorithms with respect to transmit power $P$. 
We note that SRM-CVX uses a relaxation parameter.
For fairness, we optimize both the relaxation parameter in SRM-CVX and the LogSumExp parameter $\alpha$ in our algorithms via online line searching.
In particular, we employ the online line searching method only for Fig.~\ref{fig:SRM_comp}.
In the rest of the simulations, we use fixed empirical values of the LogSumExp relaxation parameter instead of the online line searching.
\begin{figure}[t] 
    \begin{center}
    \includegraphics[width=.6\linewidth]{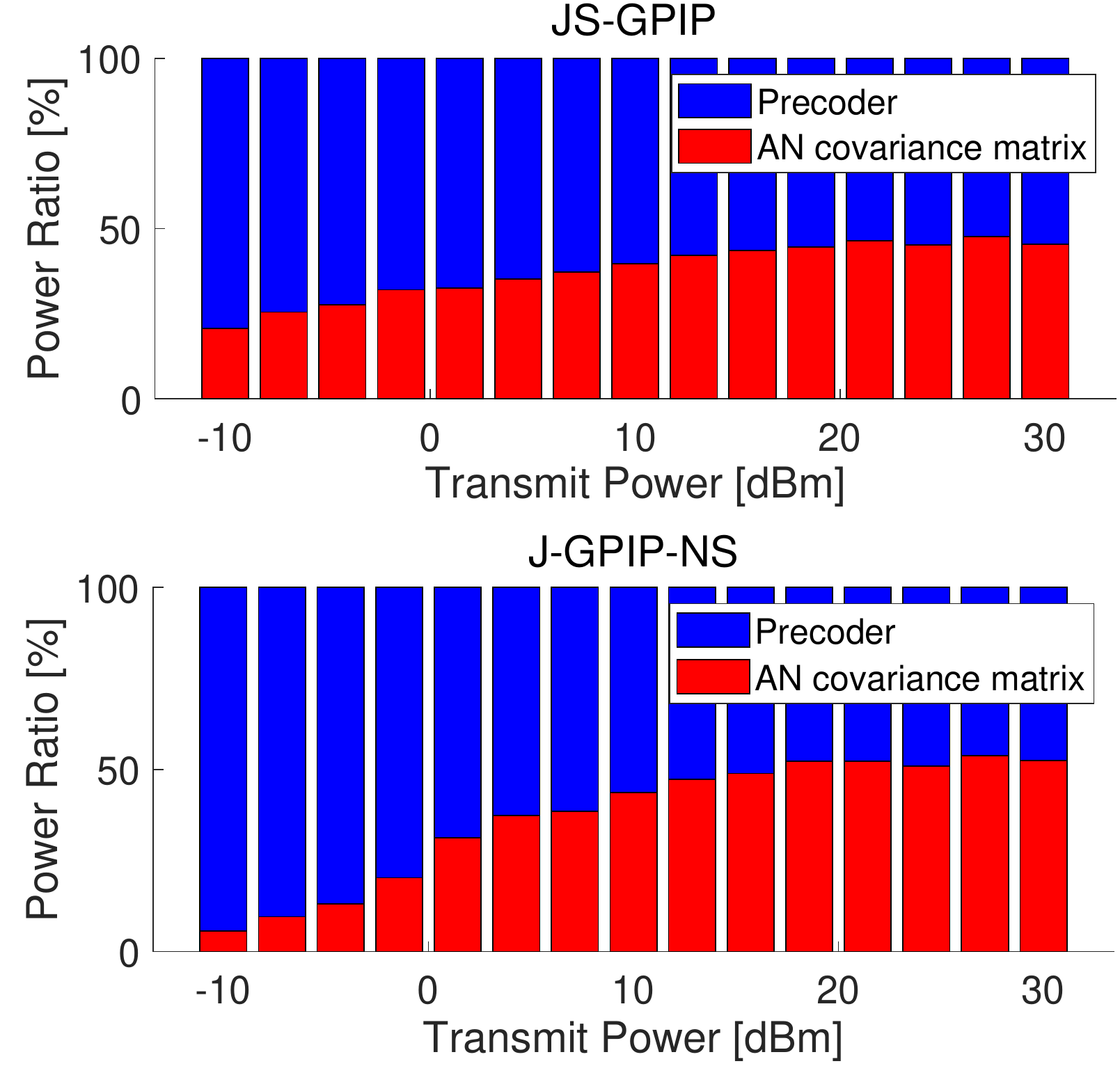}
    \end{center}
    \caption{
    Power ratio for transmit power $P$ results with $N=4$ AP antennas, $K=1$ users, and $M=3$ eavesdroppers.
    }
    \label{fig:power_ratio}
\end{figure}
As shown in Fig.~\ref{fig:SRM_comp}, JS-GPIP achieves the highest secrecy rate in the considered transmit power region followed by SRM-CVX.
In particular, MRT-NS achieves large performance gain compared to RZF-EVE-NS by focusing on maximizing the channel gain of the user and by reducing the wiretap channel rate with null-space AN.
Unlike MRT-NS precoding, RZF-EVE-NS shows relatively insufficient performance gain because the precoder of RZF-EVE-NS wastes its transmit power to reduce the wiretap channel rates, which turns out to be less efficient than concentrating the precoding power to the legitimate user in the single user case.
Since S-GPIP utilizes only the secure precoding methods without AN, it illustrates lower secrecy performance compared to JS-GPIP, J-GPIP-NS, SRM-CVX, and MRT-NS that employ AN.
Accordingly, Fig.~\ref{fig:SRM_comp}  validates the performance of JS-GPIP and demonstrates the benefit of AN by showing a significant gap between the proposed algorithms with and without AN.
 
\begin{figure}[t]
    \centering
    $\begin{array}{c c }
    {\resizebox{0.5\columnwidth}{!}
    {\includegraphics{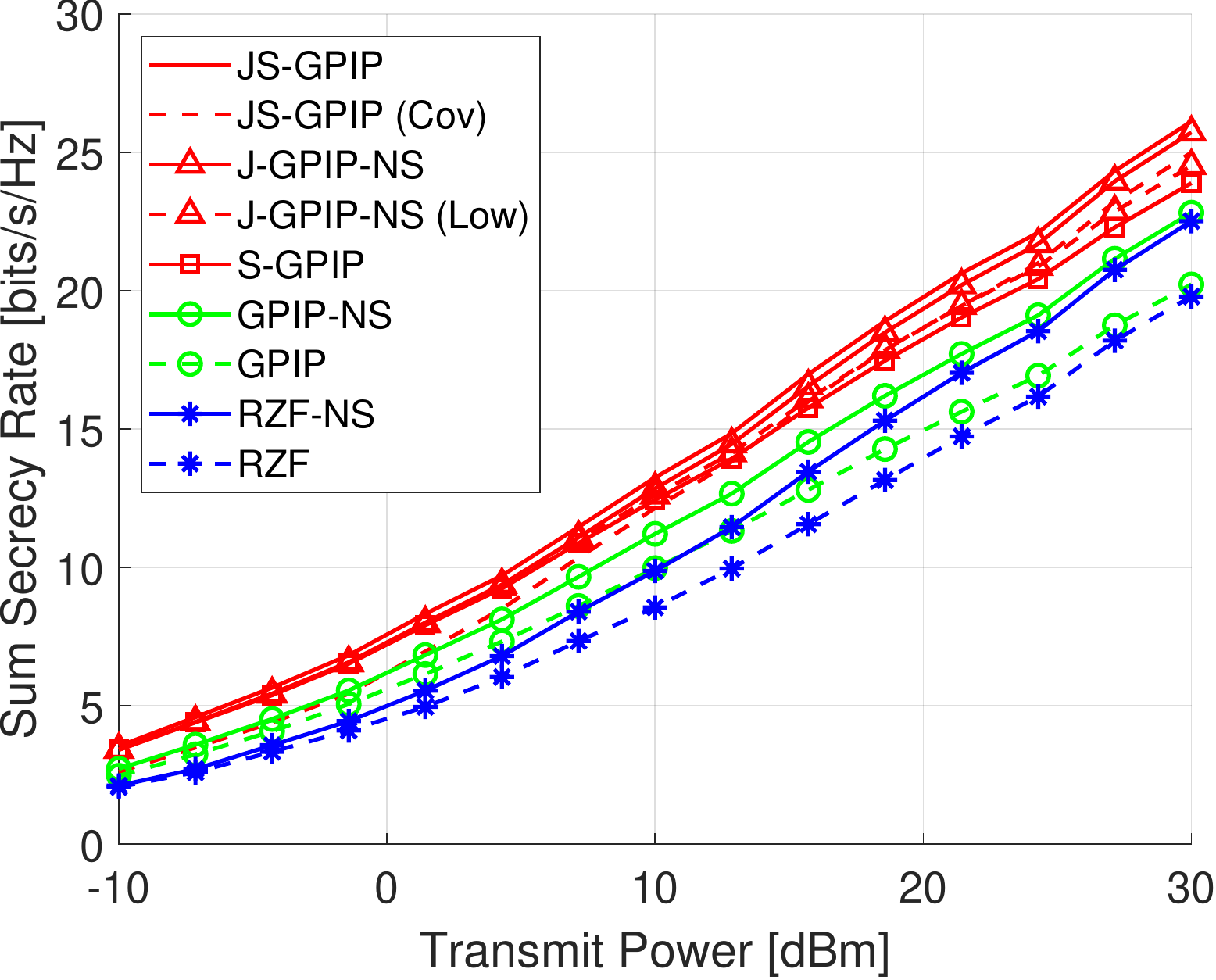}}
    } &
    {\resizebox{0.5    \columnwidth}{!}
    {\includegraphics{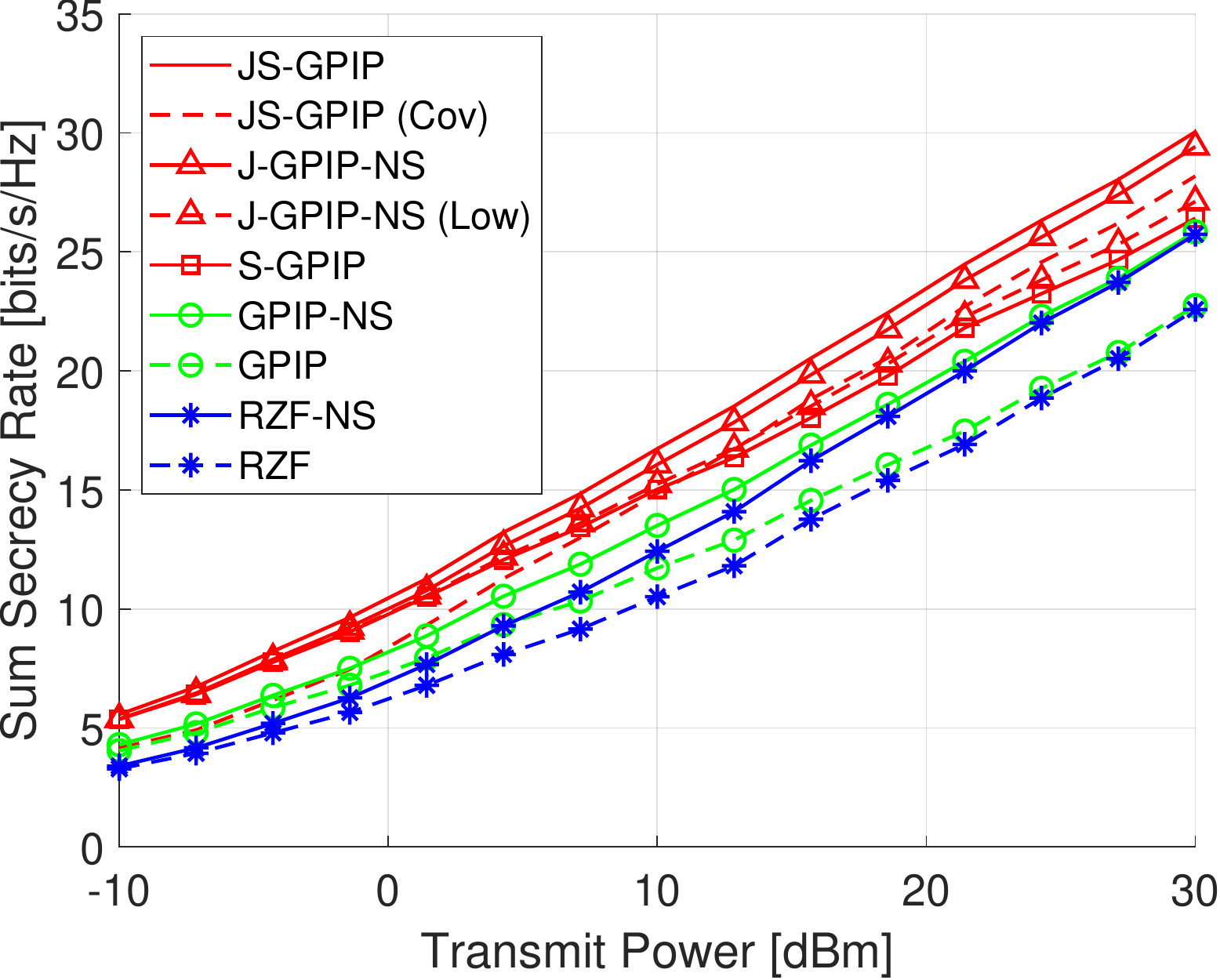}}
    }\\ \mbox{\small (a) $N = 8$ } & \mbox{\small (b) $N = 16$ }
    \end{array}$
    \caption{
    Sum secrecy rate for transmit power $P$ results with (a) $N = 8$ AP antennas, $K = 2$ users, and $M = 4$ eavesdroppers and (b) $N = 16$ AP
        antennas, $K = 2$ users, and $M = 4$ eavesdroppers.}
    \label{fig:SNR}
\end{figure}
 
In Fig.~\ref{fig:power_ratio}, we plot the power ratio between the precoder and AN with respect to the transmit power for JS-GPIP and J-GPIP-NS obtained from the case in Fig.~\ref{fig:SRM_comp}.
The power ratio of the AN increases with the transmit power, and appears to converge  in the high transmit power regime.
As the transmit power increases, the users' channel rate increases logarithmically which makes the performance gain marginal in the high transmit power regime.
Therefore, JS-GPIP and J-GPIP-NS allocate more transmit power to the AN since focusing on the wiretap channel rate reduction is more efficient strategy once the precoder is allocated with enough power in the high transmit power regime.
Although the J-GPIP-NS algorithm employs the line searching to optimize the power ratio, JS-GPIP and J-GPIP-NS illustrate similar behavior.
In particular, JS-GPIP finds the precoder, the AN covariance matrix,  and optimal power ratio simultaneously while the J-GPIP-NS attempts to find the power ratio by the line searching method for the fixed null-space AN.
In this perspective, JS-GPIP finds a better solution to maximize the sum secrecy rate, thereby showing the higher rate than that of J-GPIP-NS.
 \begin{figure}[t] 
    \begin{center}
    \includegraphics[width=.6\linewidth]{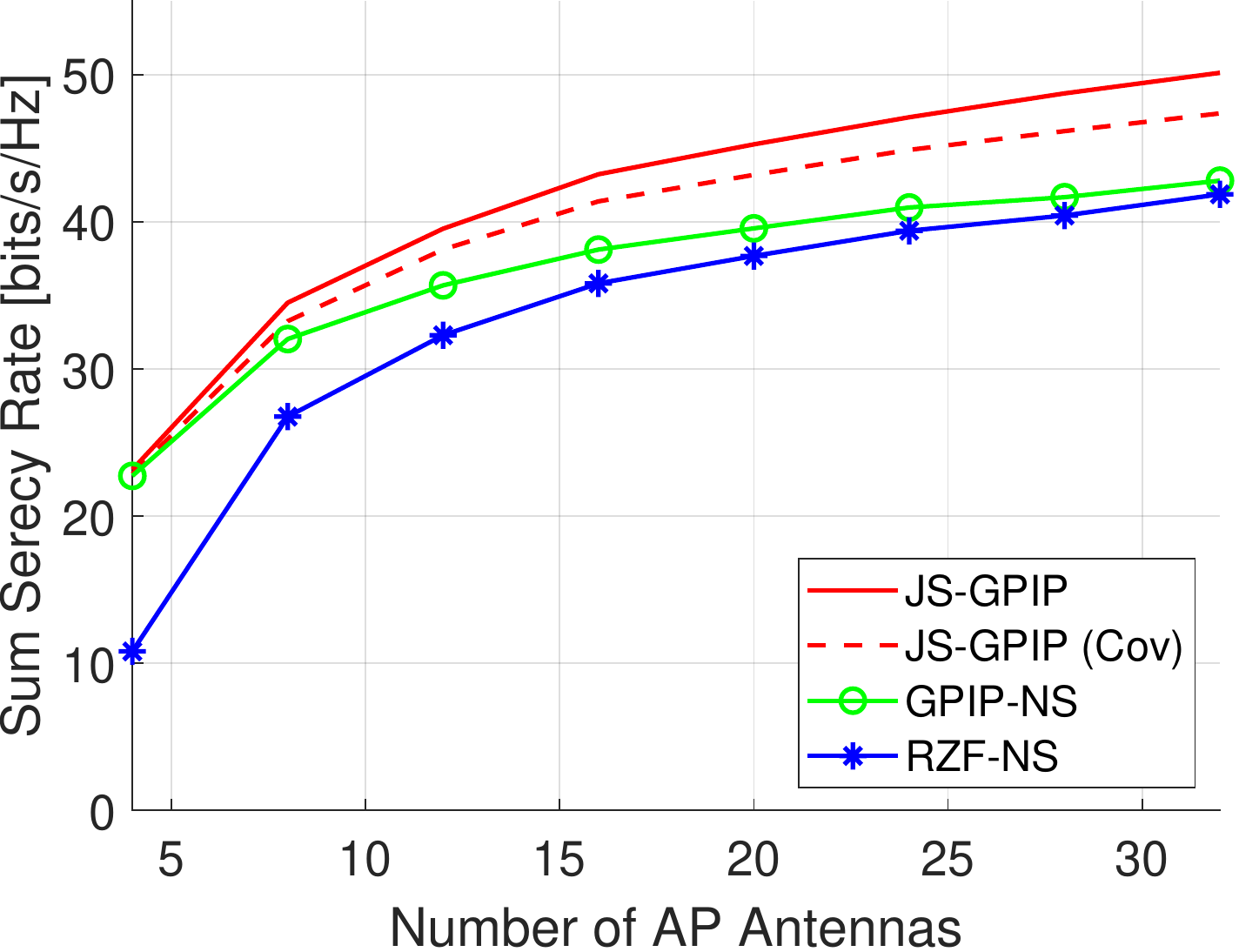}
    \end{center}
    \caption{Sum secrecy rate results for the number of AP antennas $N$ at transmit power $P=20$ dBm  with $K = 4$ users, and $M = 4$ eavesdroppers.}
    \label{fig:Antplot}
\end{figure}
Unlike SRM-CVX which works for a single user, the proposed algorithms work for a general number of users and eavesdroppers. 
In this regard, we further evaluate the proposed algorithms for multiuser cases.
In Fig.~\ref{fig:SNR}, we evaluate the sum secrecy rate of the simulated algorithms with respect to the transmit power.
We consider (a) $N=8$ AP antennas, $K=2$ users, and $M=4$ eavesdroppers and (b) $N=16$ AP antennas, $K=2$ users, and $M=4$ eavesdroppers.
In comparison to the baseline methods, the proposed algorithms achieve the highest secrecy rate  for both the cases of (a) and (b).
Since JS-GPIP jointly optimizes both the precoder and AN covariance with perfect CSIT, it provides the highest rate. 
J-GPIP-NS which also uses perfect CSIT with alternating optmization of the precoder and AN covariance shows the second highest rate.
JS-GPIP (Cov) follows J-GPIP-NS even with the partial wiretap CSIT followed by J-GPIP-NS (Low).
Finally, S-GPIP shows the lowest rate among the proposed method even with perfect CSIT as it does not utilize the AN, yet achieving the higer rate than the benchmarks.
Accordingly, such a trend and performance order are intuitive and correspond to a general insight.
In addition, the gap between GPIP and the proposed algorithms highlights the importance of the secure precoding approach.

\begin{figure}[t] 
    \begin{center}
    \includegraphics[width=.55\linewidth]{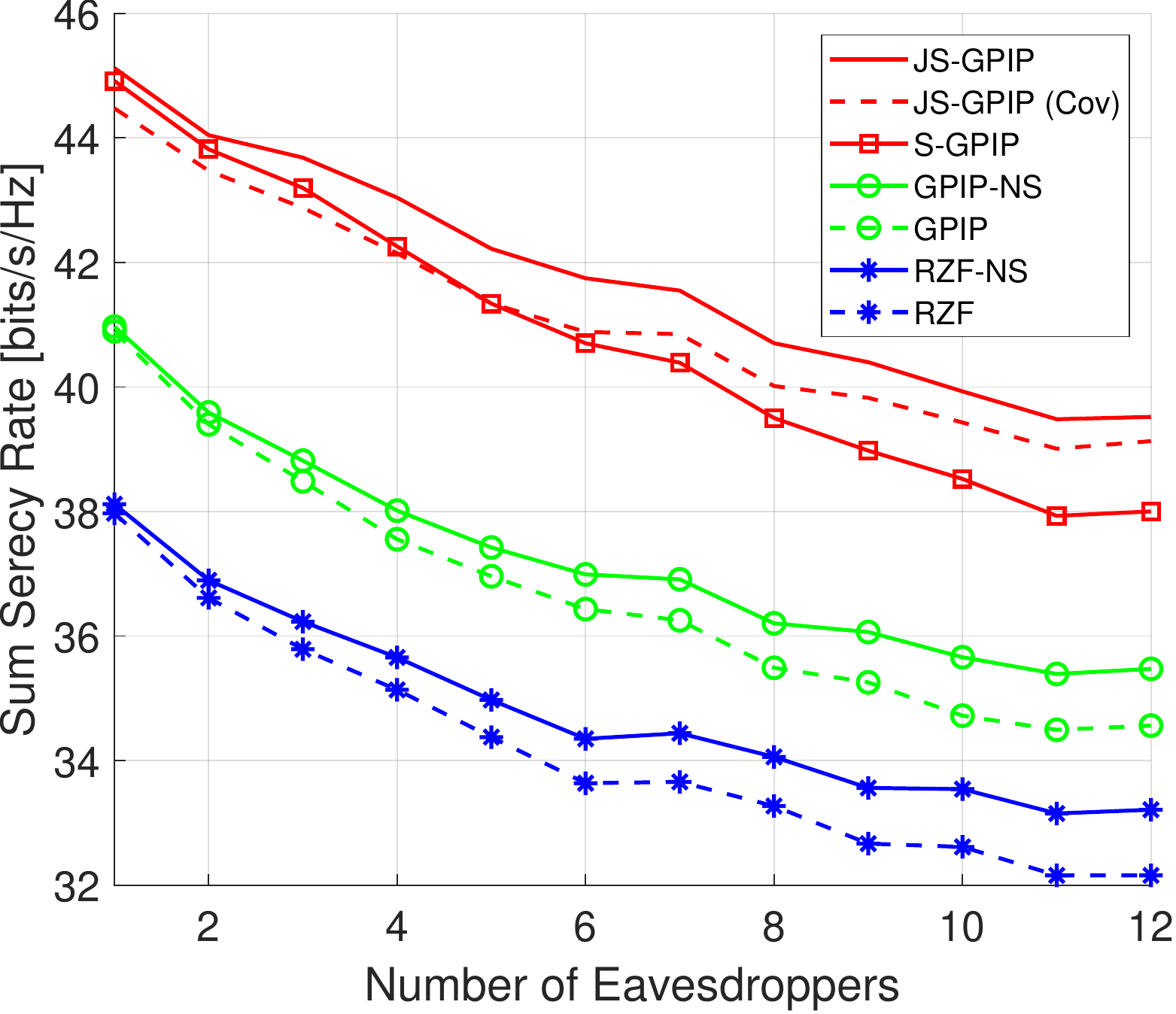}
    \end{center}
            \vspace{-1em}
    \caption{Sum secrecy rate results with respect  to the number of eavesdroppers $M$ for transmit power $P=20$ dBm, $N = 16$ AP antennas, and $K = 4$ users.}
    \label{fig:Eveplot}
            \vspace{-1em}
\end{figure}
In Fig.~\ref{fig:Antplot}, we assess the sum secrecy rate of the proposed algorithms with the other baseline algorithms in terms of the number of AP antennas $N$. 
We consider the transmit power of $P=20$ dBm,  $K=4$ users, and $M=4$ eavesdroppers.
In general, the algorithms exhibit similar performance order as in Fig.~\ref{fig:SNR}.
In particular, the proposed JS-GPIP algorithm still achieves the highest secrecy rate.
When the number of AP antennas is small, it is shown  in Fig.~\ref{fig:Antplot} that GPIP-NS achieves a similar sum secrecy rate to that of JS-GPIP because there is not enough spatial degrees of freedom to nullify the leakage channels while keeping the user SINRs reasonable.
As the number of AP antennas increases, however, the performance gap between JS-GPIP based methods and baseline methods increases; increasing the spatial degrees of freedom allows the joint optimization to be more effective.
Therefore, the proposed joint optimization can be more beneficial in the future communications where the scale of the antennas becomes large.


  
In Fig.~\ref{fig:Eveplot}, we evaluate the sum secrecy rate of the proposed algorithms with the other baseline algorithms in terms of the number of eavesdroppers $M$.
We consider $N = 16$ AP antennas, $K = 4$ users, and  $P=20$ dBm transmit power.
The proposed algorithms also achieve the highest secrecy rates.
We note that the secrecy rate gap between AN-aided and non-AN-aided algorithms increases as the number of eavesdroppers increases.
This tendency comes from the fact that the AN deteriorates all eavesdropper regardless of its number, the overall effect of using AN becomes larger with the number of eavesdroppers.
As a result, the proposed AN-aided joint secure precoding algorithm is considered to be a potential physical layer security algorithm  as the number of eavesdroppers is expected to increase in the future wireless applications.
\begin{figure}[t] 
    \begin{center}
    \includegraphics[width=.55\linewidth]{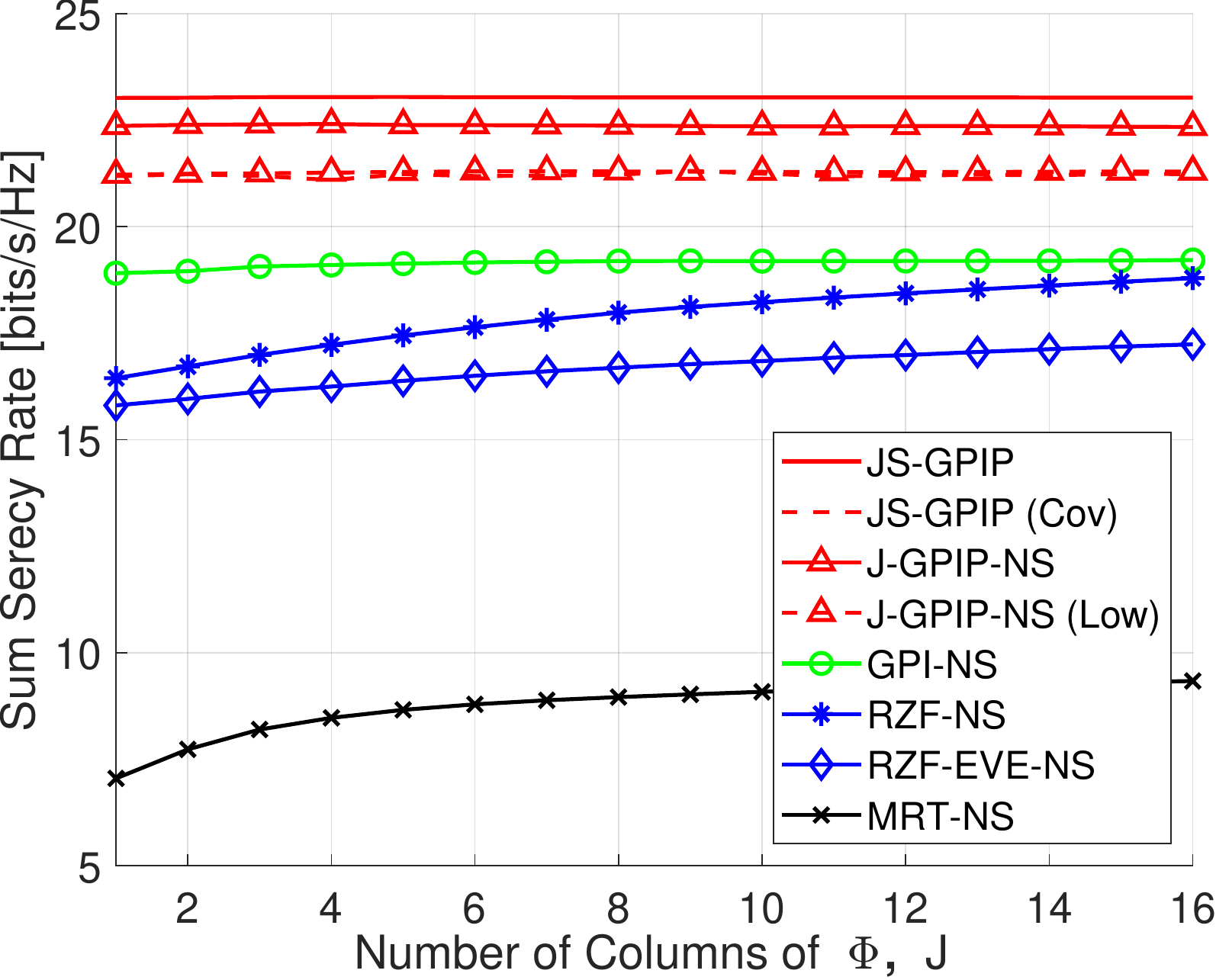}
    \end{center}
            \vspace{-1em}
    \caption{Sum secrecy rate for the number of ${\bf{\Phi}},\ J$ at transmit power $P=20$dBm results with $N = 16$ AP antennas, $K = 2$ users, and $M = 4$ eavesdroppers.}
    \label{fig:ANdimplot}
            \vspace{-1em}
\end{figure}

In Fig.~\ref{fig:ANdimplot}, we assess the sum secrecy rate of the AN-aided algorithms in terms of $J$, i.e., the number of columns of ${\bf{\Phi}}$, for $N=16$ AP antennas, $K=2$ users, $M=4$ eavesdroppers, and $P=20$ dBm transmit power. 
Unlike the benchmarks, the proposed algorithm achieves their highest secrecy rates only with $J=1$.
This result indicates that the precoder of the proposed methods can well adapt to the AN for any dimension, thereby accomplishing high efficiency with only $J=1$ dimension.
In other words, the proposed algorithms require small $J$ and thus, the total complexity of Algorithm~\ref{alg:perfectCSIT} can reduce to $\CMcal{O} \left( \frac{1}{3} T K N^3 \right)$ as discussed in Section~\ref{subsec:Alg_comp}.

In Fig.~\ref{fig:convergence}, we evaluate the convergence results in terms of the approximated objective function $L(\bar{\bv})  = \log_2\lambda(\bar{\bv})$ in \eqref{eq:lagrangian}.
Here, JS-GPIP converges within $T=5$ iterations for $P \in \{0, 20, 40\}$ dBm transmit power and   $T=10$ iterations for  $-20$ dBm transmit power.
Therefore, with the small number of iteration $T$, the proposed algorithms reveal high potential in practical implementation.  
Overall, the proposed precoding and AN covariance design methods provide improvement in secrecy rate with low complexity, and wiretap channel covariance can be utilized and enough for achieving such an improvement. 
\begin{figure}[t] 
    \begin{center}
    \includegraphics[width=.6\linewidth]{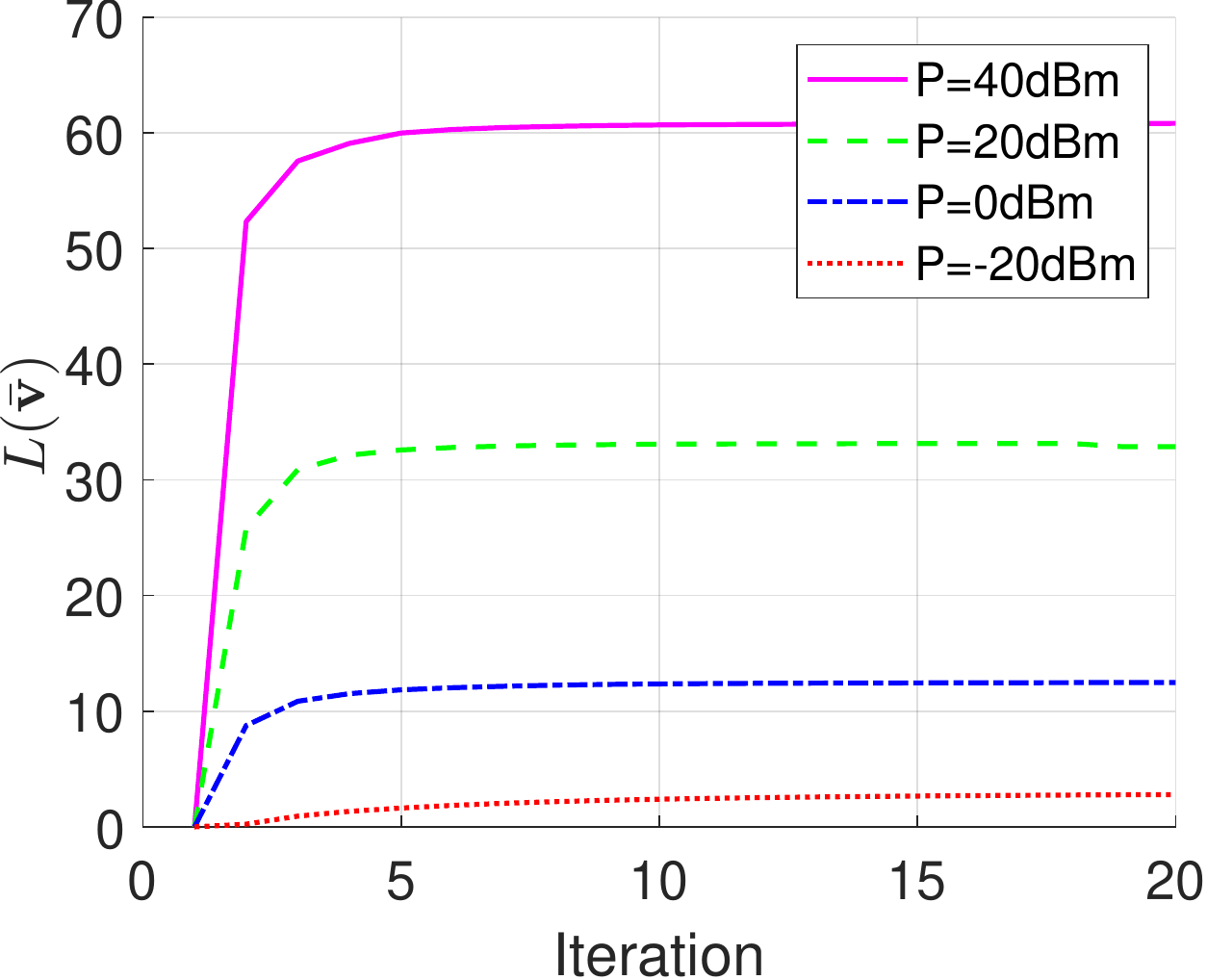}
    \end{center}
    \vspace{-1em}
    \caption{Convergence results in terms of $L(\bar \bv)$ for $N=8$ AP antennas, $K=4$ users, $M=4$ eavesdroppers, and $P=\{-20,0,20,40\}$ dBm transmit power.}
    \label{fig:convergence}
        \vspace{-1em}
\end{figure}

\section{Conclusion}
In this paper, we presented a novel precoding algorithm that maximizes the sum secrecy rate for MU-MIMO wiretap channels with multiple eavesdroppers. Our key approach was to reformulate the non-convex and non-smooth secrecy rate maximization problem into a tractable non-convex and smooth optimization problem. Then, we derived the first-order optimality condition for the reformulated problem. Leveraging this first-order condition, we proposed the secure precoding algorithms that jointly and simultaneously identify both the beams and AN covariance matrix to maximize the sum secrecy rate when perfect or imperfect channel knowledge is available at the legitimate transmitter. As a byproduct, we also proposed a secure precoding method with null-space projection-based AN design. One key observation was that by jointly and simultaneously optimizing the precoder, AN covariance, and the precoder-AN power allocation, the proposed method achieved a higher secrecy rate performance than the existing methods with the need for only a small AN subspace and with fast convergence.
Therefore, the proposed methods under the joint and simultaneous optimization framework can offer significantly improved security for future wireless applications.

\appendices
\section{Proof of Lemma 1} \label{proof:lem1}
The Lagrangian function of problem \eqref{eq:reform_problem_eve} is 
\begin{align}
    \label{eq:lagrangian}
    L(\bar {\bv}) =& \sum_{k = 1}^{K}  \left\{ \log_2  \left( \frac{{\bf \bar v }^{\sf H} {\bf{A}}_k \bar {\bv} }{{\bf \bar v }^{\sf H} {\bf{B}}_k \bar {\bv}} \right)  -  \alpha\log  \left(\sum_{m=1}^M \left( \frac{\bar{\bv}^{\sf H} {\bC}_{m,k} \bar {\bv} }{\bar{\bv}^{\sf H} {\bD}_{m,k} \bar {\bv} } \right)^{\beta}  \right) \right\}.
\end{align}
By the first-order KKT condition, a stationary point should satisfy
\begin{align} \label{eq:first_opt_cond}
\frac{\partial L(\bar {\bv})}{\partial \bar {\bv}^{\sf H}} = 0.
\end{align}
To find a condition to satisfy \eqref{eq:first_opt_cond}, we take the partial derivatives of $L(\bar {\bv})$ with respect to ${\bf \bar v }$ and set it to zero. 
For simplicity, we let 
\begin{align}
    L_1(\bar {\bv}) &=  \sum_{k = 1}^{K}  \log_2  \left( \frac{{\bf \bar v }^{\sf H} {\bf{A}}_k \bar {\bv} }{{\bf \bar v }^{\sf H} {\bf{B}}_k \bar {\bv}} \right)
    \text{ and } L_2(\bar {\bv}) = \sum_{k = 1}^{K}    \alpha\log  \left(\sum_{m=1}^M \left( \frac{\bar{\bv}^{\sf H} {\bC}_{m,k} \bar {\bv} }{\bar{\bv}^{\sf H} {\bD}_{m,k} \bar {\bv} } \right)^{\beta}  \right).
\end{align}
Then by using 
\begin{align}\label{eq:derivative_matrix}
        \frac{\partial \left( \frac{{\bf \bar v }^{\sf H} {\bA}_k \bar {\bv} }{{\bf \bar v }^{\sf H} {\bB}_k \bar {\bv}} \right)}{\partial \bar {\bv}^{\sf H}} = \left( \frac{{\bf \bar v }^{\sf H} {\bA}_k \bar {\bv} }{{\bf \bar v }^{\sf H} {\bB}_k \bar {\bv}} \right) \left[  \frac{{\bA}_k \bar {\bv} }{{\bf \bar v }^{\sf H} {\bA}_k \bar {\bv}} - \frac{{\bB}_k \bar {\bv} }{{\bf \bar v }^{\sf H} {\bB}_k \bar {\bv}} \right],
\end{align}
the partial derivative of $L_1(\bar {\bv})$ is obtained as
\begin{align}   
     \label{eq:dL1}
    \frac{\partial L_1 (\bar{\bv})}{\partial \bar {\bv}^{\sf H}} = \sum_{k=1}^K \frac{1}{\log2} \left( \frac{{\bA}_k \bar {\bv} }{{\bf \bar v }^{\sf H} {\bA}_k \bar {\bv}} - \frac{{\bB}_k \bar {\bv} }{{\bf \bar v }^{\sf H} {\bB}_k \bar {\bv}} \right).
\end{align}
Similarly, we calculate the partial derivative of $L_2(\bar {\bv})$ as
\begin{align}
    \label{eq:dL2}
    \frac{\partial L_2 (\bar{\bv})}{\partial \bar {\bv}^{\sf H}} = \sum_{k=1}^K \frac{\alpha \sum_{m} \beta \left( \frac{\bar{\bv}^{\sf H} {\bC}_{m,k} \bar {\bv} }{\bar{\bv}^{\sf H} {\bD}_{m,k} \bar {\bv} } \right)^{\beta} \left( \frac{{\bC}_{m,k} \bar {\bv} }{{\bf \bar v }^{\sf H} {\bC}_{m,k} \bar {\bv}} - \frac{{\bD}_{m,k} \bar {\bv} }{{\bf \bar v }^{\sf H} {\bD}_{m,k} \bar {\bv}} \right)}{- \sum_{m} \left( \frac{\bar{\bv}^{\sf H} {\bC}_{m,k} \bar {\bv} }{\bar{\bv}^{\sf H} {\bD}_{m,k} \bar {\bv} } \right)^{\beta}}.
\end{align}
Using \eqref{eq:dL1} and \eqref{eq:dL2} and rearrange the equations we obtain the following first order KKT condition:
\begin{align}
    \nonumber
    &\sum_{k=1}^{K}  \left(  \frac{1}{\log2} \frac{\bA_k}{ {\bf \bar v }^{\sf H} \bA_k \bar {\bv}} +  \frac{\alpha \sum_m \beta \left( \frac{\bar{\bv}^{\sf H} {\bf{C}}_{m,k} \bar {\bv} }{\bar{\bv}^{\sf H} {\bf{D}}_{m,k} \bar {\bv} }  \right)^{\beta} \left( \frac{\bD_{m,k}}{{\bf \bar v }^{\sf H} {\bD}_{m,k} \bar {\bv}} \right) }{ \sum_m  \left( \frac{\bar{\bv}^{\sf H} {\bf{C}}_{m,k} \bar {\bv} }{\bar{\bv}^{\sf H} {\bf{D}}_{m,k} \bar {\bv} }  \right)^{\beta} } \right) \bar {\bv}
    \\
    \label{eq:kkt_cond_rearrange}
    &= \sum_{k=1}^{K} \left( \frac{1}{\log2} \frac{\bB_k}{ {\bf \bar v }^{\sf H} {\bB}_k \bar {\bv}} + \frac{\alpha \sum_m \beta \left( \frac{\bar{\bv}^{\sf H} {\bf{C}}_{m,k} \bar {\bv} }{\bar{\bv}^{\sf H} {\bf{D}}_{m,k} \bar {\bv} }  \right)^{\beta} \left( \frac{\bC_{m,k}}{{\bf \bar v }^{\sf H} \bC_{m,k} \bar {\bv}} \right) }{ \sum_m  \left( \frac{\bar{\bv}^{\sf H} {\bf{C}}_{m,k} \bar {\bv} }{\bar{\bv}^{\sf H} {\bf{D}}_{m,k} \bar {\bv} }  \right)^{\beta} } \right) \bar {\bv}.
\end{align}
Finally, we further rearrange \eqref{eq:kkt_cond_rearrange} and derive 
\begin{align}
    \label{eq:KKT_final}
    {\bf{A}}_{\sf KKT}(\bar {\bv}) \bar {\bv} = \lambda(\bar {\bv}) {\bf{B}}_{\sf KKT} (\bar {\bv}) \bar {\bv}.
\end{align}
Since $\bB_{\sf KKT}$ is Hermitian, \eqref{eq:KKT_final} can be cast to \eqref{eq:kkt_lem} in Lemma~\ref{lem:main}.
This completes the proof.
\qed

\bibliographystyle{IEEEtran}
\bibliography{SecureAN}

\end{document}